\documentclass[10pt,draftclsnofoot]{IEEEtran}
\onecolumn
\usepackage{amsmath, amssymb,amsthm,cite}
\usepackage[dvips]{graphics}
\usepackage{graphicx}
\usepackage{psfrag}
\usepackage{bbm}
\usepackage{subfigure}
\usepackage{dblfloatfix}
\usepackage{mathtools}
\usepackage{algpseudocode}
\usepackage{algorithm}
\usepackage{slashbox}
\usepackage{hyperref}
\usepackage{color}

\DeclareMathOperator*{\nn}{\nonumber}

\allowdisplaybreaks
\newcommand{\RNum}[1]{\uppercase\expandafter{\romannumeral #1\relax}}
\newtheorem{lemma}{Lemma}
\newtheorem{theorem}{Theorem}
\newtheorem{corollary}{Corollary}
\theoremstyle{definition}
\newtheorem{remark}{Remark}

\newtheorem{definition}{Definition}
\makeatletter
\def\blfootnote{\gdef\@thefnmark{}\@footnotetext}
\makeatother

\begin{document}
\title{A Single-Letter Upper Bound on the Feedback Capacity of Unifilar Finite-State Channels}
\author{Oron Sabag, Haim H. Permuter and Henry D. Pfister}


\maketitle
\begin{abstract}
An upper bound on the feedback capacity of unifilar finite-state channels (FSCs) is derived. A new technique, called the $Q$-contexts, is based on a construction of a directed graph that is used to quantize recursively the receiver's output sequences to a finite set of \textit{contexts}. For any choice of $Q$-graph, the feedback capacity is bounded by a single-letter expression, $C_\text{fb}\leq \sup I(X,S;Y|Q)$, where the supremum is over $P_{X|S,Q}$ and the distribution of $(S,Q)$ is their stationary distribution. It is shown that the bound is tight for all unifilar FSCs where feedback capacity is known: channels where the state is a function of the outputs, the trapdoor channel, Ising channels, the no-consecutive-ones input-constrained erasure channel and for the memoryless channel. Its efficiency is also demonstrated by deriving a new capacity result for the dicode erasure channel (DEC); the upper bound is obtained directly from the above expression and its tightness is concluded with a general sufficient condition on the optimality of the upper bound. This sufficient condition is based on a fixed point principle of the BCJR equation and, indeed, formulated as a simple lower bound on feedback capacity of unifilar FSCs for arbitrary $Q$-graphs. This upper bound indicates that a single-letter expression might exist for the capacity of finite-state channels with or without feedback based on a construction of auxiliary random variable with specified structure, such as $Q$-graph, and not with i.i.d distribution. The upper bound also serves as a non-trivial bound on the capacity of channels without feedback, a problem that is still open.
\end{abstract}
\section{Introduction}\label{sec:intro}
\blfootnote{The work of O. Sabag and H. H. Permuter was partially supported by the European Research Council (ERC) starting grant and the Joint UGC-ISF research grant. This paper will be presented at the 2016 IEEE International Symposium on Information Theory, Barcelona, Spain. O. Sabag and H. H. Permuter are with the department of Electrical and Computer Engineering, Ben-Gurion University of the Negev, Beer-Sheva, Israel (oronsa@post.bgu.ac.il, haimp@bgu.ac.il). H. D. Pfister is with the department of Electrical and Computer Engineering, Duke University, Durham, USA (henry.pfister@duke.edu).}
A finite-state channel (FSC) is a mathematical model for channels with memory that has been applied to wireless communications and magnetic recording. In this model, the channel memory is encapsulated in a state which takes values from a finite set. A FSC is described by a state-dependent channel and a transition probability of the channel state conditioned on the input, the output and the previous channel state. In this paper, we focus on unifilar FSCs with feedback, as described in Fig. \ref{fig:FSC}, where the new channel state is a time-invariant function of the previous state, the current input and the current output.

The feedback capacity of FSCs has been investigated in \cite{PermuterWeissmanGoldsmith09,DaboraGoldIndecomposable,TatikondaMitter_IT09} and still has no closed form expression. For the special case of unifilar FSCs, it was shown in \cite{PermuterCuffVanRoyWeissman08} that the feedback capacity is:
\begin{align}\label{eq:intro_cap}
C_{\text{fb}} &= \lim_{N\rightarrow \infty} \sup_{\{p(x_t|s_{t-1},y^{t-1})\}_{t=1}^N} \frac{1}{N} \sum_{i=1}^N   I(X_i,S_{i-1};Y_i|Y^{i-1}).
\end{align}
As can be seen from the capacity formula, this capacity expression is very hard to compute in a straightforward manner. However, it was shown in \cite{PermuterCuffVanRoyWeissman08,TatikondaMitter_IT09} that the capacity can be formulated as a dynamic programming (DP) optimization problem; this has benefits such as efficient algorithms for estimating the capacity and analytical tools for calculating capacity.

The relationship between the feedback capacity of FSCs and DP first appeared in Tatikonda's thesis \cite{Tatikonda00}. The need for this formulation arises from difficulties in the computability of the capacity expression as can be seen in \eqref{eq:intro_cap}. In \cite{Chen05}, a DP formulation of a sub-family of unifilar FSCs was given, where the state can be computed at the decoder. It was shown that the DP can be analytically solved under mild conditions on the channel, resulting in a computable capacity expression. DP formulations of feedback capacities appeared also for channels where the state is determined by the inputs \cite{Yang05}, Markov channels \cite{TatikondaMitter_IT09} and Gaussian channels with stationary noise \cite{YangKavcicTatikondaGaussian}.

A typical approach for solving DP problems is the well-known Bellman equation. Loosely speaking, one should find a constant and a function which satisfy some fixed point equation; the constant is then the optimal reward (equivalent to the feedback capacity). This approach led to explicit capacity expressions for the trapdoor channel \cite{PermuterCuffVanRoyWeissman08}, the Ising channel \cite{Ising_channel,Ising_artyom}, the input-constrained erasure channel \cite{Sabag_BEC} and the input-constrained binary symmetric channel \cite{sabag_allerton_15}. The difficulty in the Bellman equation based approach lies in finding the function that satisfies this equation.

\begin{figure}[t]
\centering
    \psfrag{E}[][][.95]{Encoder}
    \psfrag{D}[][][.95]{Decoder}
    \psfrag{C}[b][][.9]{$p(y_t|x_t,s_{t-1})$}
    \psfrag{F}[t][][.9]{$s_t\mspace{-5mu}=\mspace{-5mu}f(y_t,x_t,s_{t-1})$}
    \psfrag{V}[][][.78]{Unit-Delay}
    \psfrag{M}[][][1]{$m$}
    \psfrag{Y}[][][1]{$y_t$}
    \psfrag{O}[][][1]{$\hat{m}$}
    \psfrag{Z}[][][1]{$y_{t-1}$}
    \psfrag{X}[][][1]{$x_t$}
    \includegraphics[scale = 0.5]{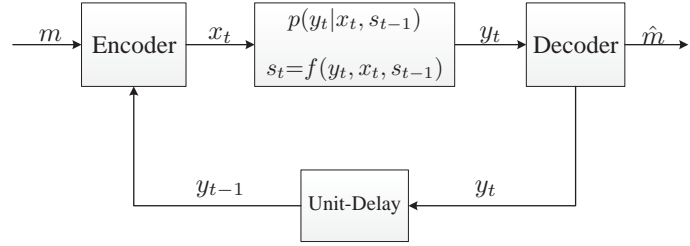}
    \caption{Unifilar FSC with feedback.}
    \label{fig:FSC}
\end{figure}

Nevertheless, computer-based simulations of DP provide bold insights into the feedback capacity expression. Specifically, implementation of the value iteration algorithm together with DP simulation give an analytic estimation of the capacity in quite a straightforward manner. This estimated value is, indeed, a lower bound on the feedback capacity, so it then remains to provide the corresponding upper bound. Therefore, in order to bypass the Bellman equation approach, one should find an alternative tool for calculating upper bounds on the feedback capacity. Our main result is a derivation of an upper bound on the feedback capacity of unifilar FSCs.

The derivation of the upper bound is initiated with an almost trivial inequality which is given for some deterministic mapping $\Phi_{i-1}$:
\begin{align*}
H(Y_i|Y^{i-1})&= H(Y_i|Y^{i-1},\Phi_{i-1}(Y^{i-1}))\\
              &\leq H(Y_i|\Phi_{i-1}(Y^{i-1})).
\end{align*}
The equality follows from the fact that $\Phi_{i-1}(\cdot)$ is a deterministic function of $Y^{i-1}$, and the inequality follows from the fact that conditioning reduces entropy. It happens that a naive choice of the mapping $\Phi_{i-1}$ might result in an equality in this upper bound; for example, if the mapping $\Phi_{i-1}$ returns a constant then the above inequality reveals the familiar converse for the capacity of memoryless channels. For non-trivial FSCs, this naive choice results an upper bound which is not tight. Throughout this paper, it will be shown that a structured mapping might improve the upper bound performance.


Our derivation is based on a new technique called the $Q$-contexts where the set of mappings $\{\Phi_i\}_{i>1}$ transforms the history of the output process into a Markov chain. Specifically, the $Q$-contexts is a mapping which is described by a directed graph where each node corresponds to a \textit{context}, and the outgoing edges per node are labelled with all possible channel outputs. Then, given an initial context (a node on the graph), a walk on the graph results in a unique mapping of output sequences onto nodes. Thus, the $Q$-graph describes a quantization of the output sequences into $Q$-context sequences.

Utilization of the $Q$-contexts technique for the capacity expression \eqref{eq:intro_cap} leads to our main result, a single-letter upper bound on the feedback capacity of unifilar FSCs:
\begin{align}\label{eq:intro_bound}
C_{\text{fb}}\leq \sup_{P_{X|S,Q}} I(X,S;Y|Q),
\end{align}
for all $Q$-graphs, where the distribution of $(S,Q)$ is determined by their stationary distribution. It is shown that the upper bound is tight for all unifilar FSCs where the feedback capacity is already known and, obviously, for the memoryless channel with feedback. Therefore, the derived upper bound also provides a unified expression for all feedback capacities known so far. This result provides hope that the feedback capacity of general FSCs might be characterized by a single-letter expression, which would be quite surprising.

Throughout the paper, we demonstrate that the bound is tight for a proper choice of the $Q$-graph for any channel where the state is computable at the decoder \cite{Chen05}, the trapdoor channel \cite{PermuterCuffVanRoyWeissman08}, the input-constrained binary erasure channel (BEC) \cite{Sabag_BEC} and Ising channels \cite{Ising_channel,Ising_artyom}. These derivations also serve as an easily implemented and alternative converse proof for these capacity results.

It is also demonstrated that the upper bound can be used to derive new capacity results, such as the dicode erasure channel (DEC). The DEC is a quantized version of the known dicode channel with additive white Gaussian noise (AWGN), which was studied in \cite{PfitserLDPC_memory_erasure,PfitserTurbo_MAP}. DP-based simulations for this channel show that the optimal policy only visits a small (i.e., finite) subset of the state space and the actions associated with those states are unconstrained. Since actions are unconstrained, the solution of the Bellman equation is very challenging, if not impossible. However, since the visited state space is finite, it is possible to extract a $Q$-graph from this simulation and to derive a simple upper bound on the capacity. Its tightness is then concluded with a sufficient condition that is derived for the upper bound \eqref{eq:intro_bound}.

The sufficient condition is based on an invariant-property of the BCJR equation for the channel state estimation, $p(s_t|y^t)$. Roughly speaking, the condition states that if a $Q$-graph and some input distribution $P_{X|S,Q}$ induce that the state estimation depends on the context, $\Phi_t(y^t)$, and not on the sequence $y^t$, then $I(X,S;Y|Q)$ is an achievable rate. This condition is easy to verify using the BCJR forward-recursive equation for unifilar FSCs and may be exploited in two ways: the first is a verification that a certain upper bound is tight, as is done for the DEC, while the second is a simple and calculable lower bound for an arbitrary $Q$-graph and input $P_{X|S,Q}$ that satisfy the condition.

The remainder of the paper is organized as follows. Section \ref{sec:prelimi} includes notation definitions and required preliminaries. Section \ref{sec:main} states our main result on the upper bound and the sufficient condition for the tightness of this bound. In Section \ref{sec:examples}, several examples of unifilar FSCs are studied and it is shown that the upper bound is tight. In Section \ref{sec:proof}, we provide a detailed proof of our main result and, finally, the paper is concluded in Section \ref{sec:conclusions}.

\section{Notation and Preliminaries}\label{sec:prelimi}
Random variables are denoted by upper-case letters, such as $X$, while realizations are denoted by lower-case letters, e.g., $x$. Calligraphic letters, e.g., $\mathcal{X}$, denote sets. We use $X^{n}$ to denote the $n$-tuple $(X_{1},\dots,X_{n})$ and $x^n$ to denote vectors of $n$ elements, i.e., $x^n = (x_1, x_2, ..., x_n)$. The binary entropy is denoted by $H_2(\alpha)=-\alpha\log_2\alpha-(1-\alpha)\log_2(1-\alpha)$, where $\alpha\in[0,1]$. Finally, $H_3(\alpha_1,\alpha_2)=-\alpha_1\log_2\alpha_1-\alpha_2\log_2\alpha_2-(1-\alpha_1-\alpha_2)\log_2(1-\alpha_1-\alpha_2)$ denotes the ternary entropy function for scalars $\alpha_1,\alpha_2\in[0,1]$ satisfying $\alpha_1+\alpha_2 \leq 1$. The quadrature entropy function, $H_4(\alpha_1,\alpha_2,\alpha_3)$, is defined in a similar manner.
\subsection{Unifilar state channels}\label{app:unifilar}
A \textit{finite state channel} is defined by the triplet $(\mathcal{X}\times\mathcal{S},p(y,s|x,s') ,\mathcal{Y}\times\mathcal{S})$ where $X$ is the channel input, $Y$ is the channel output, $S'$ is the channel state at the beginning of the transmission and $S$ is the channel state at the end of the transmission. The cardinalities of $X,Y,S$ are assumed to be finite. At each time $t$, the channel has the memoryless property
\begin{equation*}
  p(s_t,y_t|x^t,s^{t-1},y^{t-1})=p(s_t,y_t|x_t,s_{t-1}).
\end{equation*}
An FSC is called \textit{unifilar} if the new channel state, $s_t$, is a time-invariant function of the triplet $s_t=f(x_t,y_t,s_{t-1})$.

The input to the channel at time $t$, $x_t$, depends both on the message $m$ and on the output tuple $y^{t-1}$. A unifilar channel is \textit{strongly connected} if for all $s,s'\in\mathcal{S}$, there exist $T$ and $\{p(x_t|s_{t-1})\}_{t=1}^T$ such that $\sum_{t=1}^T p(S_t=s|S_0=s') > 0$. It is also assumed that the initial state, $s_0$, is available to both the encoder and the decoder.

The capacity of the unifilar FSC is given by the following theorem:
\begin{theorem}\label{theorem:capacity_unifilar}[Theorem $1$, \cite{PermuterCuffVanRoyWeissman08}]
The feedback capacity of a strongly connected unifilar state channel, where $s_0$ is available to both to the encoder and the decoder, can be expressed by
\begin{align*}
C_{\text{fb}} &= \lim_{N\rightarrow \infty} \sup_{\{p(x_t|s_{t-1},y^{t-1})\}_{t=1}^N} \frac{1}{N} \sum_{i=1}^N   I(X_i,S_{i-1};Y_i|Y^{i-1}).
\end{align*}
\end{theorem}
\subsection{Directed graphs}
A directed graph is defined by three sets of nodes, edges and labels. Node $i$ is said to \textit{communicate} with node $j$ if there exists a path from $i$ to $j$. This definition leads to an equivalence relation: two nodes $i,j$ lie in the same \textit{communicating class} if $i$ communicates with $j$ and vice versa. A communicating class is said to be closed if there are no outgoing edges from this class. A graph is \textit{irreducible} if all nodes in the graph communicate.

For a closed communicating class, the \textit{period} of a node is defined as the $\gcd$ of all natural numbers, $n$, such that there is a loop to this node with length $n$. It can be shown that the period is a class property, that is, all nodes in a closed communicating class have equal periods. A closed class is \textit{aperiodic} if it has a period of $1$.

One useful property of irreducible graphs with period $D$ is that the graph can be partitioned uniquely into $D$ disjoint subsets $A_0, A_1,\dots,A_{D-1}$ on a cycle, i.e., all edges from $A_i$ lead to $A_{(i+1)\mod D}$.
\subsection{$Q$-contexts mapping}\label{subsec:main_contexts}
The upper bound in this paper is based on the inequality:
\begin{align*}
H(Y_i|Y^{i-1})&\leq H(Y_i|\Phi_{i-1}(Y^{i-1})), \ \ i\in\mathbb{N},
\end{align*}
which holds for any set of mappings $\Phi_{i-1}:\mathcal{Y}^{i-1}\rightarrow \mathcal{Q}$. The \textit{context} of the sequence $y^{i-1}$ is defined as $q_{i-1} \triangleq \Phi_{i-1}(y^{i-1})$.

Our interest is limited to the set of mappings which can be described by a time-invariant function $g:\mathcal{Q}\times \mathcal{Y}\rightarrow\mathcal{Q}$, where $\Phi_i (y^i) = g(\Phi_{i-1}(y^{i-1}),y_i)$ for all $i$. The \textit{$Q$-contexts} mapping is defined by a function $g(\cdot,\cdot)$ or, equivalently, by a $Q$-graph with $|\mathcal{Q}|$ nodes, each taking a realization $q\in Q$; an edge $q\rightarrow q'$ with label $y$ exists if $q'=g(q,y)$. It is assumed that the $Q$-graph is finite and irreducible. These definitions imply that each node in the $Q$-graph has $|\mathcal{Y}|$ outgoing edges. An example for a $Q$-graph is illustrated in Fig. \ref{fig:BEC_Q}.
\begin{figure}[t]
\centering
    \psfrag{Q}[][][1]{$y=1$}
    \psfrag{E}[][][1]{$y=0$}
    \psfrag{F}[][][1]{$y=?$}
    \psfrag{O}[][][1]{$y=0/?/1$}
    \psfrag{L}[][][1]{$q_2$}
    \psfrag{H}[][][1]{$q_1$}
    \includegraphics[scale = 0.5]{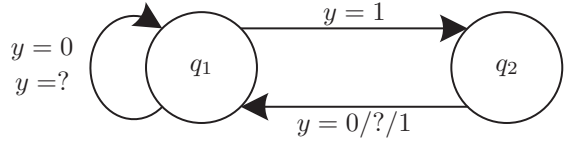}
    \caption{An example for a $Q$-graph with $|\mathcal{Q}|=2$, and $\mathcal{Y}=\{0,1,?\}$.}
    \label{fig:BEC_Q}
\end{figure}

The following step is to embed the FSC characterization into the $Q$-graph. This is done by constructing a new directed graph which includes the entire information on the $Q$-graph and the channel states evolution. A \textit{$(S,Q)$-coupled graph} is constructed as follows:
\begin{enumerate}
  \item Each node in the $Q$-graph is split into $|\mathcal{S}|$ and represented by a pair $(s,q)\in\mathcal{S}\times\mathcal{Q}$.
  \item An edge $(s,q)\rightarrow(s',q')$  with a label $(x,y)$ exists if and only if there exists a pair $(x,y)$ such that $s'=f(s,x,y)$, $q'=g(q,y)$, and $p(y|x,s)>0$.
\end{enumerate}

The coupled graph might have more than a single closed communicating class. It is then clear that if an initial pair $(s_0,q_0)$ lies in some closed communicating class, than all other classes will never be reached. Recall that $s_0$ is given by the problem, while the initial context $q_0$ is subject to any choice. The following lemma formalizes a few properties of the $(S,Q)$-graph that simplify our analysis.
\begin{lemma} \label{lemma:coupled}
There exists at least one closed communicating class in the coupled graph. For every $s\in\mathcal{S}$ (or $q\in\mathcal{Q}$) and for every closed communicating class, $\mathcal{C}$, there exists $q\in\mathcal{Q}$ (or $s\in\mathcal{S}$) such that $(s,q)\in\mathcal{C}$.
\end{lemma}
The proof of Lemma \ref{lemma:coupled} appears in Appendix \ref{app:proof_lemma_coupled}. The freedom of choosing $q_0$, together with Lemma \ref{lemma:coupled}, verifies that for a given $s_0$ there always exists $q_0$ such that $(s_0,q_0)$ lies within any of the closed classes. Therefore, we will assume throughout this paper that the $(S,Q)$-graph has a single closed communicating class only. There is no concrete example where the initial closed class effects the upper bound, but one should be aware that if a different upper bound is resulted for different closed classes, then each value is a upper bound.

In order to present the $(S,Q)$-graph as a Markov chain on $\mathcal{S}\times\mathcal{Q}$, probabilities should be assigned on the edges. For a given input matrix $P_{X|S,Q}$, an outgoing edge from $(s,q)$ that is labelled by $(x,y)$ will have a probability of $p(y|x,s)p(x|s,q)$. This assignment might effect the structure of the $(S,Q)$-graph; specifically, if an edge has $p(x|s,q)=0$ then it can be removed. Denote by $\mathcal{A}(P_{X|S,Q})$ the $(S,Q)$-graph after edge removal and define
\begin{equation}\label{eq:p_pi}
\mathcal{P}_{\pi}\triangleq\{P_{X|S,Q}: \mathcal{A}(P_{X|S,Q}) \text{ has a single closed class}\}.
\end{equation}
The subscript $\pi$ emphasizes that for all $P_{X|S,Q}\in\mathcal{P}_\pi$ there exists a unique stationary distribution on the $(S,Q)$-graph. This stationary distribution can be calculated directly on the induced single closed class, as all other nodes are inessential and have zero probability.
\section{Main result}\label{sec:main}
The following theorem is our main result:
\begin{theorem}\label{theorem:main}
The feedback capacity of a strongly connected unifilar state channel, where $s_0$ is available to both the encoder and the decoder, is bounded by
\begin{align}\label{eq:main}
C_{\text{fb}}\leq \sup_{P_{X|S,Q}\in\mathcal{P}_{\pi}}I(X,S;Y|Q),
\end{align}
for all irreducible $Q$-graphs with $q_0$ such that $(s_0,q_0)$ lies in an aperiodic closed communicating class. The joint distribution is $P_{Y,X,S,Q}=P_{Y|X,S}P_{X|S,Q}\pi_{S,Q}$, where $\pi_{S,Q}$ is the stationary distribution of the $(S,Q)$-graph.
\end{theorem}
\begin{remark}
One can view $Q$ as an auxiliary random variable (RV) representing the common knowledge that is shared by the encoder and decoder. Here, the implied sequence of auxiliary RVs has memory induced by the archive structure of the chosen $Q$-graph. In contrast, auxiliary RVs are typically chosen to be independent and identically distributed (i.i.d.), e.g., Wyner-Ziv and Gelfand-Pinsker models.

The upper bound holds for all irreducible $Q$-graphs (which satisfy the aperiodicity assumption), while in standard derivations of upper bounds and capacities it is shown that auxiliary RVs exist. Therefore, there is no optimization on this RV and cardinality bound is not relevant here since it simply holds for all $Q$-graphs. Indeed, if it can be shown that optimal $Q$-graphs always have a finite number of graph nodes, then \eqref{eq:main} will be a single-letter capacity formula for unifilar FSCs with feedback by adding a minimization over all possible $Q$-graphs.
\end{remark}
\begin{remark}
The restriction on the input distributions in $\mathcal{P}_\pi$ implies that a unique stationary distribution exists. Note that the stationary distribution $\pi_{S,Q}$ depends on the value of $P_{X|S,Q}$, and can be found as the solution of linear equations $\pi_{S,Q} T [P_{X|S,Q}]=\pi_{S,Q}$, where $T [P_{X|S,Q}]$ is the transfer matrix of the $(S,Q)$-graph as a function of $P_{X|S,Q}$.
\end{remark}
\begin{remark}
As discussed in Section \ref{sec:prelimi}, if the $(S,Q)$-graph contains more than a single closed class, then the upper bound holds for all closed communicating classes which are aperiodic. This fact follows from Lemma \ref{lemma:coupled}, where it is shown that each closed class contains all $s_0\in\mathcal{S}$ and all $q_0\in\mathcal{Q}$.
\end{remark}
\begin{remark}
Since the transmitter is free to ignore the feedback, the feedback capacity is greater than or equals the non-feedback capacity. Thus, Theorem \ref{theorem:main} also provides a computable and non-trivial upper bound on the non-feedback capacity of a unifilar FSC, which is still an open problem.
\end{remark}
\begin{remark}
An efficient method for finding the optimal $Q$-contexts is to study the corresponding DP. Standard simulations (see \cite{PermuterCuffVanRoyWeissman08,Ising_channel,Sabag_BEC}) produce a histogram of the DP states that are visited under an estimated optimal policy. The inaccuracy of such simulations follows from the required quantization of the DP parameters.

When the resulting histogram of the DP states is discrete, i.e., only a finite number of DP states are visited, then a $Q$-graph can be extracted from the DP simulation. Specifically, each visited DP state is taken as a node in the $Q$-graph and the labelled edges are the evolution of the DP states as a function of the outputs.
\end{remark}
In the following section, a sufficient condition for the optimality of the upper bound is provided.
\subsection{Lower bound on capacity}
Before presenting the lower bound, the BCJR recursive equation of the channel state estimation is derived: for an outputs tuple $y^t$ and a state $s_t\in\mathcal{S}$,
\begin{align}\label{eq:BCJR}
p(s_t|y^t)&= \frac{p(s_t,y_t|y^{t-1})}{p(y_t|y^{t-1})}  \nn\\
          &= \frac{\sum_{x_{t},s_{t-1}} p(s_t,y_t,x_{t},s_{t-1}|y^{t-1})}{\sum_{x_{t},s_{t-1}} p(y_t,x_{t},s_{t-1}|y^{t-1})}.
\end{align}
This is a forward-recursive equation in the sense that with a set of scalars $\{p(S_{t-1}=s|y^{t-1})\}_{s\in\mathcal{S}}$ and an output symbol, $y_t$, one can compute the set $\{p(S_{t}=s|y^{t})\}_{s\in\mathcal{S}}$. Note that the collection of scalars $p(s_t|y^t)$ is an element from the simplex of size $|\mathcal{S}|$, denoted here by $\mathcal{Z}$.

Given an irreducible $Q$-graph, an input distribution $P_{X|S,Q}\in\mathcal{P}_\pi$ is said to be an \textit{aperiodic input} if its corresponding $(S,Q)$-graph is aperiodic. Therefore, for each aperiodic input one can write the BCJR equation in \eqref{eq:BCJR} as a mapping $B:\mathcal{Z}\times\mathcal{Y}\to\mathcal{Z}$. Finally, each aperiodic input induces some stationary distribution $\pi_{S,Q}$, and we say that an aperiodic input is \textit{BCJR-invariant} if the set of conditional stationary distributions, $\{\pi_{S|Q=q}\}_{q\in Q}$, satisfies
\begin{align*}
  \pi_{S|Q=g(q,y)}&=B(\pi_{S|Q=q},y),
\end{align*}
for all $q\in Q$ and $y\in\mathcal{Y}$ where $g(q,y)$ is the context that is calculated from the node $q$ and the output $y$.

The following theorem provides a lower bound on feedback capacity.
\begin{theorem}\label{theorem:lower}
The feedback capacity of unifilar FSCs is bounded by
\begin{align}\label{eq:Theorem_Lower}
C_{\text{fb}}&\geq I(X,S;Y|Q),
\end{align}
for all aperiodic inputs, $P_{X|S,Q}\in\mathcal{P}_\pi$, that are BCJR-invariant.
\end{theorem}
\begin{remark}
Theorem \ref{theorem:lower} acts as a complementary tool for the upper bound in Theorem \ref{theorem:main}. One application is to evaluate the upper bound from Theorem \ref{theorem:main} for some $Q$-graph and then to verify its optimality by the BCJR-invariant property. However, there are cases where the upper bound is tight and the corresponding BCJR-invariant property is not satisfied; therefore, this property is a sufficient but not a necessary condition for the optimality of the upper bound. Nevertheless, the above statement suggests a lower bound for all aperiodic inputs, so it can be exploited as a lower bound with an arbitrary aperiodic input, as we will see in Section \ref{subsec:BEC}.
\end{remark}
\section{Examples}\label{sec:examples}
This section covers several examples from the literature where the capacity of a unifilar FSC is known.
\subsection{Channel state is a function of the outputs}
In \cite{Chen05}, a unifilar FSC where the channel state is available to all parties and evolves according to $p(s_{t+1}|x_{t},s_{t})$ was studied. It was shown that this class of FSCs is, indeed, equivalent to a unifilar FSC where the channel state is the last output, i.e., $s_i=y_i$. The authors showed that for channels with strongly irreducible and aperiodic states\footnote{A channel is strongly irreducible if the graph with $|\mathcal{S}|$ nodes (each corresponds to an output) and the edge $s' \rightarrow s$ exists if $p(s|x,s')>0$ for all $x$, is irreducible. Strong aperiodicity is defined in a similar manner.}, the capacity is given by $C_{\text{fb}} = \max_{P_{X|S}} I(X;Y|S)$, where $P_{Y,X,S}=P_{Y|X,S}P_{X|S}\pi_{S}$.

To apply Theorem \ref{theorem:main} for this case, the $Q$-graph is taken as the states graph since states can be computed from outputs. If the channel states form an aperiodic graph, then Theorem \ref{theorem:main} gives
\begin{align}\label{eq:example_chen}
  C_{\text{fb}}&\leq \sup_{P_{X|S}\in\mathcal{P}_\pi}I(X;Y|S).
\end{align}

Note that the derived bound \eqref{eq:example_chen} holds for outputs that are strongly connected and form an aperiodic graph, while in \cite{Chen05}, the outputs are assumed to be strongly irreducible and aperiodic, which is a stronger property.
\begin{figure}[h]
\centering
    \psfrag{A}[][][1]{$(q_1,s=0)$}
    \psfrag{B}[][][1]{$(q_1,s=1)$}
    \psfrag{C}[][][1]{$(q_2,s=0)$}
    \psfrag{D}[][][1]{$(q_2,s=1)$}
    \psfrag{E}[][][1]{$(0,\phi)$}
    \psfrag{F}[][][1]{$(1,1)$}
    \psfrag{G}[][][1]{$(0,\phi)$}
    \psfrag{H}[][][1]{$(1,\phi)$}
    \psfrag{I}[][][1]{$(1,?)$}
    \includegraphics[scale = 0.4]{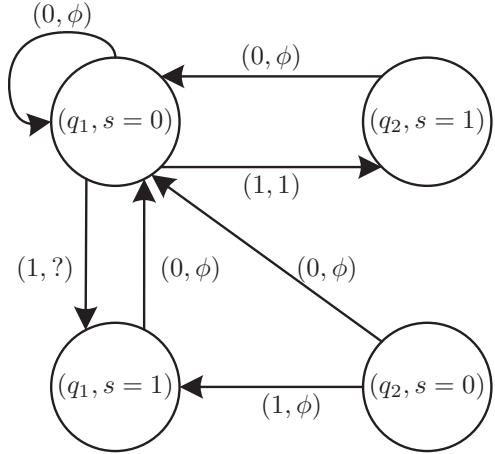}
    \caption{The $(S,Q)$-coupled graph for the input-constrained BEC for the $Q$-graph in Fig. \ref{fig:BEC_Q}. Each edge is labelled by a pair $(x,y)$, where $\phi$ stands for the ``don't care" symbol, i.e., all possible outputs.}
    \label{fig:BEC_coupled}
\end{figure}
\subsection{Input-constrained BEC}\label{subsec:BEC}
The setting consists of a BEC, where inputs must admit the $(1,\infty)$-RLL constraint, i.e., the input sequence does not contain consecutive ones. This setting does not fall into the classical definition of unifilar FSCs. However, it is possible to convert the input constraint into a channel state, $s_i=x_i$, and to derive the upper bound in Theorem \ref{theorem:main}, when the maximization is over constrained inputs. The feedback capacity of this channel was found in \cite{Sabag_BEC} using an explicit and tedious solution for the Bellman equation.

The following result is a re-statement of the known feedback capacity.
\begin{theorem}\label{theorem:BEC}[Theorem $1$, \cite{Sabag_BEC}]
The feedback capacity of the input-constrained BEC is
\begin{align}\label{eq:BEC_upper}
  C_{\text{BEC}}&=  \max_{0\leq p\leq 0.5} \frac{H_2(p)}{\frac{1}{1-\epsilon}+p}.
\end{align}
\end{theorem}
Here, we will provide an alternative proof for Theorem \ref{theorem:BEC}: the upper bound is shown by applying Theorem \ref{theorem:main} with the $Q$-graph presented in Fig. \ref{fig:BEC_Q}, while the lower bound is achieved by applying Theorem \ref{theorem:lower} with a new $Q$-graph that is presented in Fig. \ref{fig:BEC_Q_FO}.
\begin{remark}
In this example, the BCJR-invariant property is not satisfied for the graph that is presented in Fig. \ref{fig:BEC_Q} and, therefore, this is a sufficient but not a necessary condition for the tightness of the upper bound.
On the other hand, calculation of the upper bound with the $Q$-graph in Fig. \ref{fig:BEC_Q_FO} results in a tight upper bound as well; however, it is preferable to calculate the upper bound using a $Q$-graph with the fewest nodes.
\end{remark}
\begin{figure}[h]
\centering
    \psfrag{A}[][][1]{$q_1$}
    \psfrag{B}[][][1]{$q_2$}
    \psfrag{C}[][][1]{$q_3$}
    \psfrag{D}[][][1]{$y=1$}
    \psfrag{E}[][][1]{$y=0$}
    \psfrag{F}[][][1]{$y=?$}
    \includegraphics[scale = 0.6]{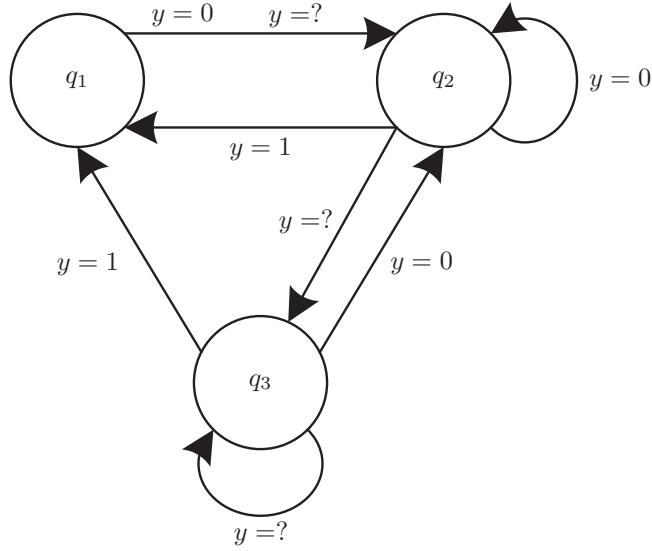}
    \caption{$Q$-graph for the input-constrained BEC with an output alphabet $\mathcal{Y}=\{0,1,?\}$.}
    \label{fig:BEC_Q_FO}
\end{figure}
\textbf{Upper bound:}
The $(S,Q)$-graph for the $Q$-graph from Fig. \ref{fig:BEC_Q} is presented in Fig. \ref{fig:BEC_coupled}. There is a single closed class in this graph consisting of all nodes except $(q_2,s=0)$, which is aperiodic since there is a loop of length $1$. Since inputs are constrained, we have $p(x=1|s=1,q)=0$ for all $q$ and, therefore, the matrix $P_{X|S,Q}$ can be parameterized with a single parameter $p(x=1|s=0,q_1)=p$.

Calculation of the stationary distribution for the $(S,Q)$-graph gives
$[\pi_{0,1}, \pi_{1,1},\pi_{1,2}] = \left[\frac{1}{1+p},\frac{\epsilon p}{1+p},\frac{(1-\epsilon) p}{1+p}\right]$, where $\pi_{i,j}=\pi(s=i,q_j)$. Then, one can calculate the conditional distribution,
\begin{align}\label{eq:BEC_V}
  p(y=1|q_1) &= p(y=1,x=1,s=0|q_1) \nn\\
  &= (1-\epsilon) p \frac{\pi_{0,1}}{\pi_{0,1} + \pi_{1,1}} \nn\\
  &= \frac{(1-\epsilon) p}{1+\epsilon p}.
\end{align}
By Theorem \ref{theorem:main}, we have:
\begin{align*}
  C_{\text{BEC}}&\leq \sup_{p(x|s,q)\in\mathcal{P}_\pi}  I(X,S;Y|Q)\\
  &= \sup_{p(x|s,q)\in\mathcal{P}_\pi} H(Y|Q) - H_2(\epsilon)  \\
  &\stackrel{(a)}= \max_{0\leq p \leq 1}  [\pi_{0,1}+ \pi_{1,1}] H_3\left(\frac{(1-\epsilon) p}{1+\epsilon p},\epsilon\right) + \pi_{1,2}H_2(\epsilon)  - H_2(\epsilon)\\
  &\stackrel{(b)}= \max_{0\leq p \leq 1}  \frac{1+\epsilon p}{1+p} (1-\epsilon)H_2\left(\frac{p}{1+\epsilon p}\right)\\
  &\stackrel{(c)}\leq \max_{0\leq p\leq 1}  \frac{H_2(p)}{\frac{1}{1-\epsilon}+p} ,
\end{align*}
where $(a)$ follows from $\pi_{0,2}=0$ and substituting \eqref{eq:BEC_V}, $(b)$ follows from the identity $H_3((1-\epsilon)x,\epsilon) = H_2(\epsilon) + (1-\epsilon)H_2(x)$, for all $x\in[0,1]$, and $\sum_{i,j}\pi_{i,j}=1$, finally, $(c)$ follows by exchanging the maximization variable to be $p'\triangleq \frac{p}{1+\epsilon p}$ and taking its maximization domain to be $[0,1]$. $\hfill\blacksquare$

\textbf{Lower bound:} consider the $Q$-graph that is presented in Fig. \ref{fig:BEC_Q_FO} with inputs that are given by $p(x=1|s=0,q_2)=p$ and $p(x=1|s=0,q_3)=\frac{p}{1-p}$ where $p\in[0,0.5]$.

Construction of the $(S,Q)$-graph reveals that the pairs $(s=0,q_1)$ and $(s=1,q_2)$ cannot be reached, while the stationary distribution of the other pairs is positive and equals:
\begin{align}\label{eq:BEC_stationary_lower}
               \pi_{1,1} &=  \frac{\bar{\epsilon} p}{1+ \bar{\epsilon} p}  \nn\\
               \pi_{0,2} &=  \frac{\bar{\epsilon}}{1+ \bar{\epsilon} p}  \nn\\
               \pi_{0,3} &=  \frac{\epsilon(1-p)}{1+ \bar{\epsilon} p}  \nn\\
               \pi_{1,3} &=  \frac{\epsilon p}{1+ \bar{\epsilon} p},
\end{align}
where $\pi_{i,j}=\pi(s=i,q_j)$. By \eqref{eq:BEC_stationary_lower}, it can be calculated that
\begin{align*}
        \pi(s=0|q_1) &= 0  \nn\\
        \pi(s=0|q_2) &= 1  \nn\\
        \pi(s=0|q_3) &= 1-p.
\end{align*}

Since $|\mathcal{S}|=2$, the value of $\pi(s=0|q_i)$ determines uniquely the value of $\pi(s=1|q_i)$ and it is sufficient to show the BCJR-invariant property for $\pi(s=0|q_i)$.
The BCJR equation can then be written as:
\begin{equation*}
p(s_i=0|q_j=g(q_i,y)) =\left\{\begin{array}{cc}
 1 & \text{if } y=0, \\
 1 - \pi(s=0|q_i) p(x=1|s=0,q_i) & \text{if } y=?, \\
0 & \text{if } y=1, \end{array}\right.
\end{equation*}

We show the BCJR-invariant property for each node: the node $q_1$ only has input edges that are labeled by $y=1$ and, therefore, $p(s_i=0|q_j=g(q_i,y=1))=\pi(s=0|q_1)=0$ for $i=2,3$, as required. For the node $q_2$, all incoming edges are labelled by $y=0$ except for the edge $q_1\to q_2$ that is labelled with $y=?$. For this edge, calculation gives that $1-\pi(s=0|q_1) p(x=1|s=0,q_1) = 1$ and it can be concluded that the node $q_2$ is BCJR-invariant as well. Finally, $q_3$ has two incoming edges that satisfy $1 - \pi(s=0|q_2) p(x=1|s=0,q_2)= 1 - \pi(s=0|q_3) p(x=1|s=0,q_3)=1-p$.

By Theorem \ref{theorem:lower},
\begin{align}\label{eq:BEC_lower_conclusion}
  C_{\text{BEC}}&\geq  I(X,S;Y|Q) \nn\\
               &= (1-\epsilon)\left[\frac{\bar{\epsilon}}{1+ \bar{\epsilon} p} H_2(p) + \frac{\epsilon}{1+ \bar{\epsilon} p} H_2(p)\right]\nn\\
               &= \frac{H_2(p)}{\frac{1}{1-\epsilon} + p}.
\end{align}
Since the lower bound \eqref{eq:BEC_lower_conclusion} holds for all $p\in[0,5]$, maximization on this parameter can be performed and concludes the proof of this theorem.$\hfill\blacksquare$

\begin{figure}[h]
\centering
    \psfrag{A}[][][1]{$0$}
    \psfrag{B}[][][1]{$1$}
    \psfrag{C}[][][1]{$0$}
    \psfrag{D}[][][1]{$?$}
    \psfrag{E}[][][1]{$1$}
    \psfrag{F}[][][1]{$\epsilon$}
    \psfrag{G}[][][1]{$1-\epsilon$}
    \psfrag{X}[][][1]{$X_i$}
    \psfrag{Y}[][][1]{$Y_i$}
    \psfrag{M}[][][1]{$-1$}
    \psfrag{P}[][][1]{$1$}
    \psfrag{Z}[][][1]{$0$}
    \psfrag{R}[][][1]{$X_{i-1}=1$}
    \psfrag{T}[][][1]{$X_{i-1}=0$}
    \includegraphics[scale = 0.45]{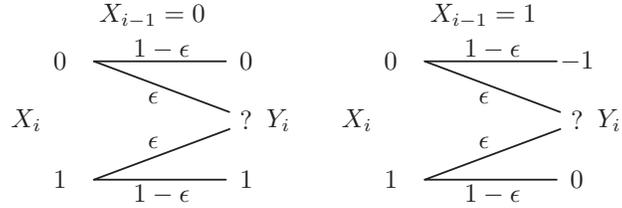}
    \caption{DEC with erasure probability $\epsilon$.}
    \label{fig:DEC}
\end{figure}
\subsection{Dicode erasure channel (DEC)}
The DEC \cite{PfitserLDPC_memory_erasure,henry_dissertation}, as described in Fig. \ref{fig:DEC}, is a simplified version of the well-known dicode channel with AWGN. Specifically, a binary input goes through a discrete-time linear filter described by $1-D$, i.e., the filter outputs $x_i-x_{i-1}$ on the real line, and this is then transmitted on an erasure channel.

Inputs are taken from $\mathcal{X}=\{0,1\}$, while outputs take values in $\mathcal{Y}=\{-1,0,1,?\}$. The channel output is $y_i=x_{i} - x_{i-1}$ with probability $1-\epsilon$, and equals $y_i=?$ with probability $\epsilon$, where $\epsilon$ is a parameter in $[0,1]$. It is evident that the DEC is a unifilar FSC if the channel state is taken as the previous input, i.e.,  $s_i = x_i$.
\begin{figure}[h]
\centering
    \psfrag{A}[][][1]{$q_1$}
    \psfrag{B}[][][1]{$q_2$}
    \psfrag{C}[][][1]{$q_3$}
    \psfrag{D}[][][1]{$y=-1$}
    \psfrag{E}[][][1]{$y=1$}
    \psfrag{F}[][][1]{$y=?$}
    \psfrag{Z}[][][1]{$y=0$}
    \includegraphics[scale = 0.6]{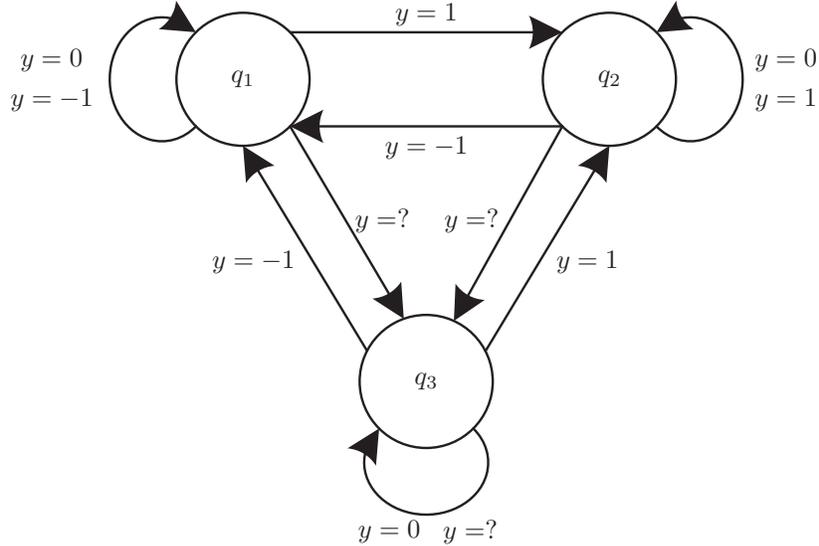}
    \caption{$Q$-graph for the DEC with an output alphabet $\mathcal{Y}=\{0,1,-1,?\}$.}
    \label{fig:DEC_Q}
\end{figure}

The following theorem encapsulates the feedback capacity for the DEC.
\begin{theorem}[DEC capacity]\label{theorem:DEC}
The feedback capacity of the DEC is:
\begin{align}\label{eq:main_DEC}
C_{\text{DEC}}&= \max_{0\leq p \leq 1} (1-\epsilon)\frac{p + \epsilon H_2(p)}{\epsilon + (1-\epsilon)p}.
\end{align}
\end{theorem}
Theorem \ref{theorem:DEC} is obtained by calculating the upper bound with the $Q$-graph from Fig. \ref{fig:DEC_Q}, and the lower bound follows from the sufficient condition provided in Theorem \ref{theorem:lower}. Indeed, this $Q$-graph has a nice interpretation as $q_1$ and $q_2$ correspond to a perfect knowledge of the channel state by the decoder, while $q_3$ implies that the decoder does not know the channel state.

\begin{proof}[Proof of Theorem \ref{theorem:DEC}]
For the $Q$-graph in Fig. \ref{fig:DEC_Q}, its corresponding $(S,Q)$-graph can be described with a matrix. Each input to the matrix corresponds to a pair $(x,y)$ of input and output, and we use $\phi$ as a notation for all possible channel outputs.
\begin{align*}
                   \begin{array}{c|c|c|c|c|c|c}
                         -    & (0,q_1) & (1,q_1)  &  (0,q_2) & (1,q_2) &   (0,q_3)  & (1,q_3) \\ \hline
                    (s=0,q_1) & (0,0)     &   -        &    -       & (1,1)     &    (0,?)     & (1,?)    \\
                    (s=1,q_1) & (0,-1)    &  (1,0)     &      -     &      -    &    (0,?)     & (1,?)    \\
                    (s=0,q_2) &    -      &    -       &    (0,0)   & (1,1)     &    (0,?)     & (1,?)    \\
                    (s=1,q_2,) & (0,-1)    &      -     &     -      & (1,0)     &    (0,?)     & (1,?)     \\
                    (s=0,q_3) &  -        &     -      &       -    & (1,1)     &    (0,\phi)  & (1,?)    \\
                    (s=1,q_3) & (0,-1)    &     -      &      -     &    -      &    (0,?)     & (1,\phi)    \\
                   \end{array}
\end{align*}
From the matrix above, it can be noted that the nodes $(s=1,q_1)$ and $(s=0,q_2)$ are not in the single closed communicating class that is formed by all other states. This closed class is aperiodic since there is a loop with length $1$ for the node $(s=0,q_1)$.

By exploiting the symmetry of the channel and the $(S,Q)$-graph, the maximization on input distributions can be limited to:
\begin{align*}
  p(x=0|q_1,s=0)&=p(x=1|q_2,s=1)= a \\
  p(x=1|q_3,s=0)&=p(x=0|q_3,s=1)= p.
\end{align*}
It follows that the stationary distribution is:
\begin{align}\label{eq:DEC_stationary}
               \pi_{0,1} &=  \frac{(1- \epsilon)p}{2 \epsilon  + 2 (1-\epsilon)p}  \nn\\
               \pi_{1,2} &=  \frac{(1- \epsilon)p}{2 \epsilon  + 2 (1-\epsilon)p}  \nn\\
               \pi_{0,3} &=  \frac{\epsilon}{2 \epsilon  + 2 (1-\epsilon)p}  \nn\\
               \pi_{1,3} &=  \frac{\epsilon}{2 \epsilon  + 2 (1-\epsilon)p},
\end{align}
where $\pi_{i,j}=\pi(s=i,q_j)$.

Consider the following chain of equalities:
\begin{align}\label{eq:DEC_HYQ}
H(Y|Q)&= \sum_{q=1}^3 (\pi_{0,q} + \pi_{1,q})H(Y|Q=q) \nn\\
&\stackrel{(a)}=\frac{(1- \epsilon)p}{\epsilon  + (1-\epsilon)p}H_3((1-\epsilon)a,(1-\epsilon)(1-a)) + \frac{\epsilon}{\epsilon  + (1-\epsilon)p}H_4\left(\epsilon,(1-\epsilon)\frac{p}{2},(1-\epsilon)\frac{p}{2}\right) \nn\\
&\stackrel{(b)}=(1-\epsilon)\left[\frac{(1- \epsilon)p}{\epsilon  + (1-\epsilon)p}H_{2}(a) + \frac{\epsilon}{\epsilon  + (1-\epsilon)p}H_3\left(\frac{p}{2},\frac{p}{2}\right)\right] + H_2(\epsilon) \nn\\
&\stackrel{(c)}=(1-\epsilon)\left[\frac{(1- \epsilon)pH_{2}(a)}{\epsilon  + (1-\epsilon)p} + \frac{\epsilon(p+H_2(p))}{\epsilon  + (1-\epsilon)p}\right] + H_2(\epsilon),
\end{align}
where $(a)$ is obtained by substituting the stationary distribution from \eqref{eq:DEC_stationary}, $(b)$ follows from the identity $H_3((1-\delta)\gamma,(1-\delta)(1-\gamma)) = H_2(\delta) + (1-\delta)H_2(\gamma)$ by choosing $\delta=\epsilon,\gamma=a$. Finally, $(c)$ follows from the above identity by choosing $1-\delta=p$ and $\gamma=\frac{1}{2}$.

The upper bound on the capacity can then be established:
\begin{align*}
C_\text{DEC}&\leq\sup_{P_{X|S,Q}\in\mathcal{P}_\pi} I(X,S;Y|Q)\\
&\stackrel{(a)}= \max_{a,p} H(Y|Q) - H(Y|X,S,Q)\\
&\stackrel{(b)}=\max_{a,p} (1-\epsilon)\left[\frac{(1- \epsilon)pH_{2}(a)}{\epsilon  + (1-\epsilon)p} + \frac{\epsilon(p+H_2(p))}{\epsilon  + (1-\epsilon)p}\right]\\
&\stackrel{(c)}=\max_{p} (1-\epsilon)\frac{ p + \epsilon H_2(p)}{\epsilon  + (1-\epsilon)p},
\end{align*}
where $(a)$ follows from $H(Y|X,S,Q)=H_2(\epsilon)$, $(b)$ follows from \eqref{eq:DEC_HYQ} and $H(Y|X,S,Q)= H_2(\epsilon)$ and $(c)$ follows from $H_2(a)\leq 1$.

For the lower bound on feedback capacity, we simply take the maximizing distribution from the upper bound $p(x=0|s=0,q_1)=p(x=1|s=1,q_2)=0.5$ and $p(x=1|s=0,q_3)=p(x=0|s=1,q_3)=p$ for some $p\in[0,1]$ and show that the BCJR-invariant property is satisfied. This input distribution is an aperiodic input since the $(S,Q)$-graph has a loop with length $1$. The stationary distribution which is given in \eqref{eq:DEC_stationary} gives that $[\pi(s=0|q_1),\pi(s=0|q_2),\pi(s=0|q_3)] = [1,0,0.5]$.

The BCJR equation can be calculated:
\begin{equation*}
p(s=0|g(q,y) )=\left\{\begin{array}{cc}
 1 & \text{if } y=-1, \\
 0 & \text{if } y=1, \\
 \pi(s=0|q)p(x=0|s=0,q)+\pi(s=1|q)p(x=1|s=0,q) & \text{if } y=?, \\
  \frac{\pi(s=0|q)p(x=0|s=0,q)}{\pi(s=0|q)p(x=0|s=0,q)+\pi(s=1|q)p(x=1|s=1,q)} & \text{if } y=0.\end{array}\right.
\end{equation*}

To show the BCJR-invariant property, it is convenient to treat each output observation separately. First, it is easy to note that all edges with $y=-1$ or $y=1$ lead to $q_1$ and $q_2$, respectively, which approves the invariant property since $\pi(s=0|q_1)=1$ and $\pi(s=0|q_2)=0$. For the output $y=?$, one can show that
\begin{align*}
  \pi(s=0|q_i)p(x=0|s=0,q_i)+\pi(s=1|q_i)p(x=1|s=0,q_i)&=0.5,
\end{align*}
for $i=1,2,3$. For the output $y=0$, the BCJR-invariant property can be established in a similar manner and this concludes that the input is BCJR-invariant. Since we used the $Q$-graph and the maximizer of the upper bound, there is no need to calculate the expression $I(X,S;Y|Q)$ since, obviously, it equals the upper bound.
\end{proof}
\begin{figure}[h]
\centering
    \psfrag{E}[][][1]{$y=1$}
    \psfrag{D}[][][1]{$y=0$}
    \psfrag{C}[][][1]{$q_3$}
    \psfrag{B}[][][1]{$q_2$}
    \psfrag{A}[][][1]{$q_1$}
    \includegraphics[scale = 0.5]{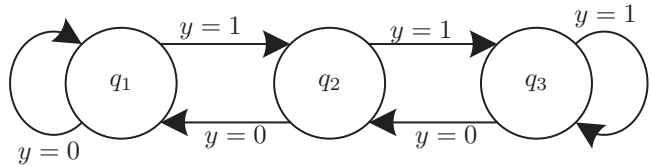}
    \caption{$Q$-contexts graph for the trapdoor channel.}
    \label{fig:Q1_trapdoor}
\end{figure}

\subsection{Trapdoor channel}\label{subsec:trapdoor}
The trapdoor channel was invented by Blackwell \cite{Blackwell_trapdoor} in 1961. The capacity of this channel has been investigated in several papers and still remained an open problem. The channel has $\mathcal{S}=\mathcal{X}=\mathcal{Y}=\{0,1\}$. The output of the channel $y_t$ is equal to $s_{t-1}$ with probability $p$ and equals $x_t$ with probability $1-p$. Here, $p$ is the channel parameter and can take any value in $[0,1]$. Finally, the channel state is $s_t=s_{t-1}\oplus x_t\oplus y_t$, where $\oplus$ is the XOR operation.

In \cite{PermuterCuffVanRoyWeissman08}, the feedback capacity for the trapdoor channel with parameter $p=0.5$ was shown to be $\log_2 \phi$, where $\phi$ is the known golden ratio. The solution relied on a DP formulation of the problem and, then, establishing a solution for the Bellman equation. As a first stage, we would like to provide an alternative converse for $p=0.5$ which simplifies their original proof.

Direct application of Theorem \ref{theorem:main} with the $Q$-graph in Fig \ref{fig:Q1_trapdoor} gives:
\begin{theorem}[Upper bound]\label{theorem:trapdoor_Q1}
The feedback capacity of the trapdoor channel is bounded by
\begin{align}\label{eq:Q1_trapdoor}
C_{\text{Trap}}(p)&\leq \max_{(\alpha_1,\alpha_2,\alpha_3)\in[0,1]^3}  2(\kappa_1+\kappa_2)H_2\left(\frac{\kappa_1(1-\alpha_1(1-p)) + \kappa_2(1-p)\alpha_2}{\kappa_1+\kappa_2}\right) \nn\\& - 2H_2(p)(\kappa_1\alpha_1 + \kappa_2\alpha_2-0.5\alpha_3) + 2\kappa_3,
\end{align}
where
\begin{align*}
  \delta &= 2(1-p)[ \alpha_1 -\alpha_2 + \alpha_1\alpha_3 - \alpha_1\alpha_2 + \alpha_2\alpha_3] + 4\alpha_1p - 2\alpha_3 + 2 \\
  \kappa_1 &= \frac{(1-\alpha_3)(1-\alpha_2 (1-p) )}{\delta} \\
  \kappa_2 &= \frac{\alpha_1(p + \alpha_3 (1-p))}{\delta} \\
  \kappa_3 &= \frac{\alpha_1(1- \alpha_2(1-p))}{\delta}.
\end{align*}
\end{theorem}
The proof of Theorem \ref{theorem:trapdoor_Q1} is omitted and follows by direct application of Theorem \ref{theorem:main} with the $Q$-graph from Fig. \ref{fig:Q1_trapdoor}. A special case of Theorem \ref{theorem:trapdoor_Q1} is when $p=0.5$ and careful calculation gives
\begin{corollary}[Upper bound, $p=0.5$]\label{coro:trapdoor_05}
The feedback capacity of the trapdoor channel with $p=0.5$ is bounded by
\begin{align}\label{eq:trapdoor_05}
C_{\text{Trap}}(0.5)&\leq \log_2 \phi.
\end{align}
\end{corollary}
Therefore, it follows that the upper bound from Theorem \ref{theorem:trapdoor_Q1} is tight for $p=0.5$. Note that, at this point, the tightness of the upper bound follows from our previous knowledge of the feedback capacity in \cite{PermuterCuffVanRoyWeissman08}. Next, we use Theorem \ref{theorem:lower} to show that $\log_2 \phi$ is achievable not only for $p=0.5$ but for all $p\in[0,1]$.
\begin{theorem}[Lower bound]\label{theorem:trapdoor_lower}
The feedback capacity of the trapdoor channel is bounded by
\begin{align}\label{eq:trapdoor_lower}
  C_{\text{Trap}}(p)&\geq \log_2\phi,
\end{align}
for all $p$.
\end{theorem}
Corollary \ref{coro:trapdoor_05} and Theorem \ref{theorem:trapdoor_lower} provide an alternative proof for the feedback capacity presented in \cite{PermuterCuffVanRoyWeissman08}. The proofs of Corollary \ref{coro:trapdoor_05} and Theorem \ref{theorem:trapdoor_lower} appear in Appendix \ref{app:coro_trapdoor_05} and Appendix \ref{app:trapdoor_lower}, respectively.

\begin{figure}[h]
\centering
        \psfrag{A}[][][1.2]{$q_1$}
        \psfrag{B}[][][1.2]{$q_2$}
        \psfrag{C}[][][1.2]{$q_3$}
        \psfrag{D}[][][1.2]{$q_4$}
        \psfrag{E}[][][0.9]{$y=0$}
        \psfrag{G}[][][1]{$y=1$}
        \centerline{\includegraphics[scale = 0.5]{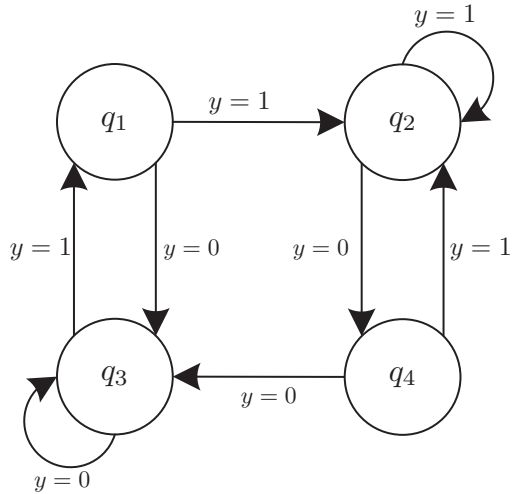}}
    \caption{Expanded $Q$-contexts graph for the trapdoor channel.}
\label{fig:Q2_trapdoor}
\end{figure}

Let us extend our realm of interest to a general parameter; numerical evaluation of Theorem \ref{theorem:trapdoor_Q1} and a lower bound that is obtained from DP simulations give the plotted results in Fig. \ref{fig:trapdoor_simulations}. Coarse inspection shows that the upper bound and the lower bound do not coincide in general, except for when $p=0.5$. Now, an \textit{expanded $Q$-graph} is introduced in Fig. \ref{fig:Q2_trapdoor} and is plotted in Fig. \ref{fig:trapdoor_simulations} with the same lower bound from DP simulations. It can be seen that the new upper bound shows a significant improvement in comparison with the upper bound in Fig. \ref{fig:trapdoor_simulations}.


\begin{figure}[!h]
\begin{center}
\subfiguretopcaptrue
\subfigure[]{
    \includegraphics[scale = 0.35]{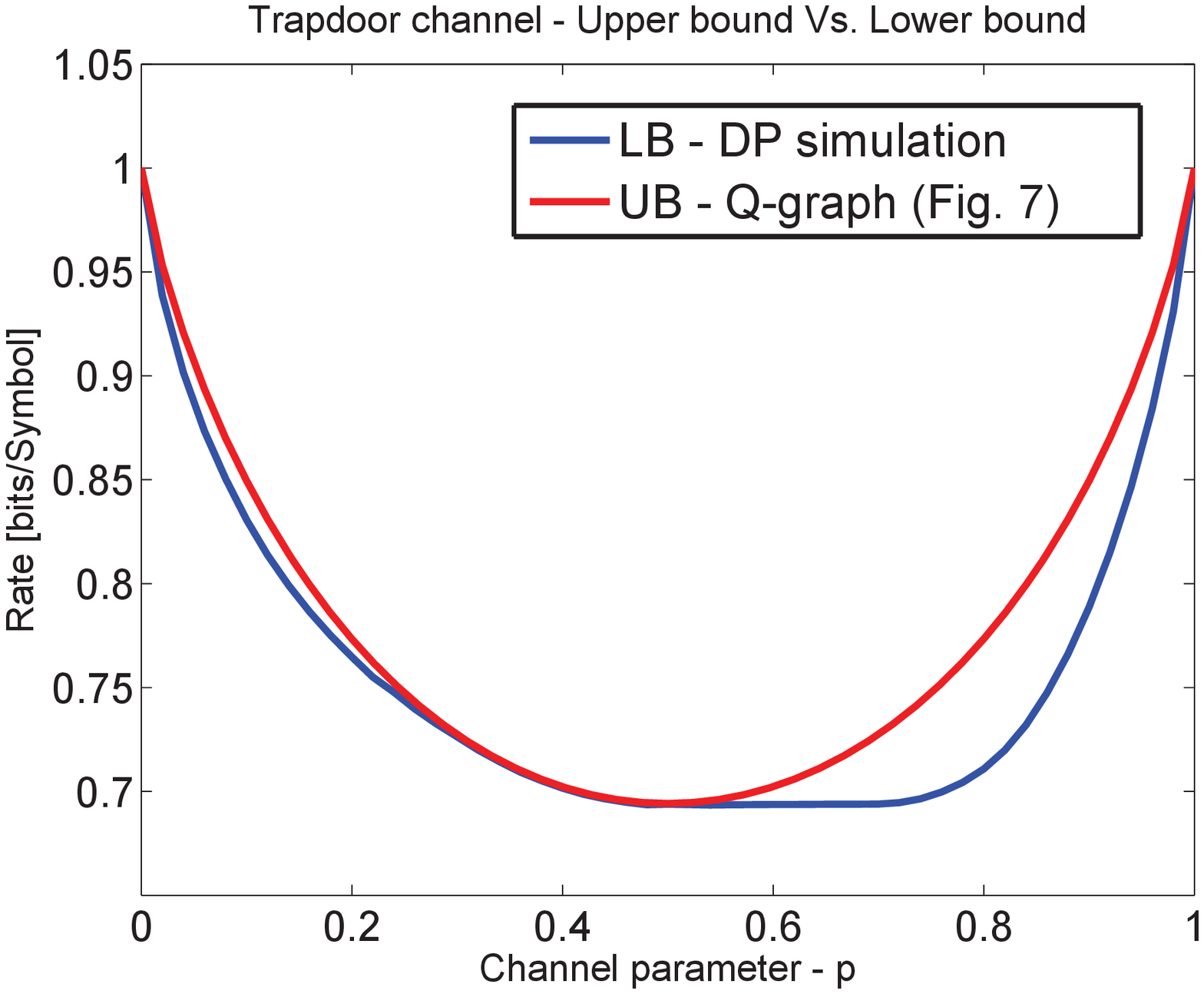}}
\subfigure[]{
    \includegraphics[scale = 0.35]{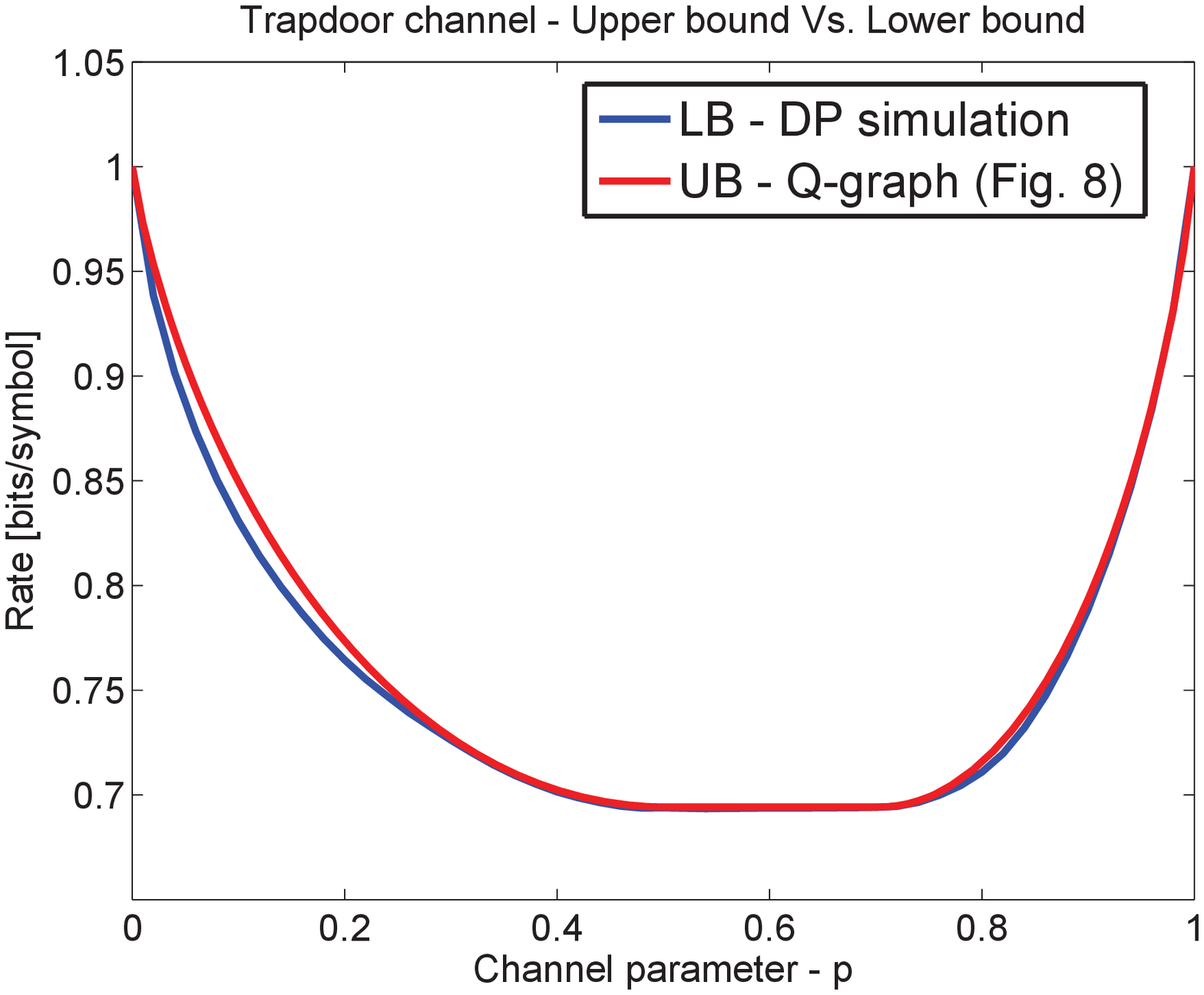}}
\caption{A comparison between a lower bound (LB) on the feedback capacity that is achieved from DP simulation, and two upper bounds that were obtained from Theorem \ref{theorem:main}. In $(a)$, the upper bound is calculated  with the $Q$-graph (Fig. \ref{fig:Q1_trapdoor}), while in $(b)$ the upper bound is calculated with the expanded $Q$-graph (Fig. \ref{fig:Q2_trapdoor}).}
\label{fig:trapdoor_simulations}
\end{center}
\end{figure}

\section{Proof of Theorem \ref{theorem:main}}\label{sec:proof}
An outline of the proof of Theorem \ref{theorem:main} is given here and comprises three building blocks appearing in Lemmas \ref{lemma:step1} - \ref{lemma:step3}. The first step expresses the essence of our bound and is encapsulated in the following lemma:
\begin{lemma}[Step $1$]\label{lemma:step1}
For a strongly connected unifilar state channel, where $s_0$ is available to both the encoder and the decoder,
\begin{align}\label{eq:step1}
  C_{\text{fb}}&\leq \sup_{\{p(x_t|s_{t-1},q_{t-1})\}_{t\geq 1}} \liminf_{N\rightarrow \infty}\frac{1}{N} \sum_{i=1}^N I(X_i,S_{i-1};Y_i|Q_{i-1}),
\end{align}
for all $Q$-contexts. The joint distribution is calculated with respect to
\begin{align*}
  p(s_0,q_0)\cdot\prod_{i=1}^N p(x_i|s_{i-1},q_{i-1}) p(y_i|x_i,s_{i-1})\mathbbm{1}_{s_i=f(s_{i-1},x_i,y_i)}\mathbbm{1}_{q_i=g(q_{i-1},y_i)}.
\end{align*}
\end{lemma}
The proof of Lemma \ref{lemma:step1} appears in Section \ref{subsec:step1}.

The upper bound in Lemma \ref{lemma:step1} is still difficult to compute since it is given by a limiting expression. The second step of the proof is tedious but necessary for our derivation, as we show that it is sufficient to restrict our maximization domain to the stationary inputs distribution. This step relies heavily on the DP formulation of the upper bound in \eqref{eq:step1}; then, a known result from the literature is used to show the existence of an optimal stationary policy (equivalent to stationary inputs distribution). This second step is precisely outlined as follows:
\begin{lemma}[Step $2$]\label{lemma:step2}
It is sufficient to maximize the upper bound in \eqref{eq:step1} over stationary input distributions, i.e.,
\begin{align}\label{eq:step2}
\sup_{\{p(x_t|s_{t-1},q_{t-1})\}_{t\geq 1}} \liminf_{N\rightarrow \infty}\frac{1}{N} \sum_{i=1}^N I(X_i,S_{i-1};Y_i|Q_{i-1})&= \sup_{P_{X|S,Q}} \liminf_{N\rightarrow \infty}\frac{1}{N} \sum_{i=1}^N I(X_i,S_{i-1};Y_i|Q_{i-1}),
\end{align}
for all irreducible $Q$-graphs with $q_0$ that lies in an aperiodic closed class. The input distribution in the RHS of \eqref{eq:step2} is $P_{X|S,Q}$ at all times.
\end{lemma}
The proof of Lemma \ref{lemma:step2} appears in Section \ref{subsec:step2}.

Finally, the calculation of the upper bound with stationary inputs can be made; a minor restriction on the maximization domain verifies the existence of a stationary distribution on the $(S,Q)$-graph and, then,
\begin{lemma}[Step $3$]\label{lemma:step3}
\begin{align}\label{eq:step3}
\sup_{P_{X|S,Q}} \liminf_{N\rightarrow \infty}\frac{1}{N} \sum_{i=1}^N I(X_i,S_{i-1};Y_i|Q_{i-1}) &\leq \sup_{P_{X|S,Q}\in\mathcal{P}_{\pi}}I(X,S;Y|Q),
\end{align}
where $\mathcal{P}_{\pi}$ is defined in \eqref{eq:p_pi}. If the supremum is attained with an aperiodic input then \eqref{eq:step3} holds with equality.
\end{lemma}
The proof of Lemma \ref{lemma:step3} appears in Section \ref{subsec:step3}.

\subsection{Proof of Lemma \ref{lemma:step1} (Step $1$):}\label{subsec:step1}

The proof comprises of the following steps:
\begin{align*}
C_{\text{fb}} &\stackrel{(a)}=     \lim_{N\rightarrow \infty} \sup_{\{p(x_t|s_{t-1},y^{t-1})\}_{t=1}^N} \frac{1}{N} \sum_{i=1}^N   I(X_i,S_{i-1};Y_i|Y^{i-1})                 \\
       &              =     \lim_{N\rightarrow \infty} \sup_{\{p(x_t|s_{t-1},y^{t-1})\}_{t=1}^N} \frac{1}{N} \sum_{i=1}^N   H(Y_i|Y^{i-1}) - H(Y_i|X_i,S_{i-1})     \\
       &\stackrel{(b)}\leq  \lim_{N\rightarrow \infty} \sup_{\{p(x_t|s_{t-1},y^{t-1})\}_{t=1}^N} \frac{1}{N} \sum_{i=1}^N   H(Y_i|Q_{i-1}) - H(Y_i|X_i,S_{i-1},Q_{i-1})\\
       &              =     \lim_{N\rightarrow \infty} \sup_{\{p(x_t|s_{t-1},y^{t-1})\}_{t=1}^N} \frac{1}{N} \sum_{i=1}^N   I(X_i,S_{i-1};Y_i|Q_{i-1})\\
       &\stackrel{(c)}=     \lim_{N\rightarrow \infty} \sup_{\{p(x_t|s_{t-1},q_{t-1})\}_{t=1}^N} \frac{1}{N} \sum_{i=1}^N   I(X_i,S_{i-1};Y_i|Q_{i-1})\\
       &\stackrel{(d)}= \sup_{\{p(x_t|s_{t-1},q_{t-1})\}_{t\geq 1}} \liminf_{N\rightarrow \infty}\frac{1}{N} \sum_{i=1}^N I(X_i,S_{i-1};Y_i|Q_{i-1}).
\end{align*}
where
 \begin{itemize}
   \item[(a)] follows from Eq. $(18)$ in Theorem $1$ \cite{PermuterCuffVanRoyWeissman08};
   \item[(b)] follows from the fact that conditioning reduces entropy and from the Markov chain of the channel $Y_i-(X_i,S_{i-1})-Y^{i-1} - \Phi_{i-1}(Y^{i-1})\triangleq Q_{i-1}$;
   \item[(c)] follows from Lemma \ref{lemma:domain};
   \item[(d)] follows from the arguments in \cite[Lemma 4]{PermuterCuffVanRoyWeissman08}.
 \end{itemize}
$\hfill\blacksquare$
\begin{lemma}\label{lemma:domain}
Given $(s_0,q_0)$, the maximization domain can be restricted as
\begin{align}\label{eq:main_symmetric}
  \sup_{\{p(x_t|s_{t-1},y^{t-1})\}_{t=1}^N} \sum_{i=1}^N   I(X_i,S_{i-1};Y_i|Q_{i-1})&=\sup_{\{p(x_t|s_{t-1},q_{t-1})\}_{t=1}^N} \sum_{i=1}^N   I(X_i,S_{i-1};Y_i|Q_{i-1}).
\end{align}
for all $N$.
\end{lemma}
\begin{proof}[Proof of Lemma \ref{lemma:domain}]\label{proof_lemma_domain}
It is shown that the same objective is achieved when exchanging the domain $\{p(x_t|s_{t-1},y^{t-1})\}_{t=1}^N$ with the domain $\{p(x_t|s_{t-1},q_{t-1})\}_{t=1}^N$; the second domain is calculated as the marginal distribution of $\{p(x_t,s_{t-1},y^{t-1})\}_{t=1}^N$ that is induced by the first domain. To this end, we show that two distributions $\{p_1(x_t|s_{t-1},y^{t-1})\}_{t=1}^N$ and $\{p_2(x_t|s_{t-1},y^{t-1})\}_{t=1}^N$  with the same induced marginal distribution $\{\tilde{p}(x_t|s_{t-1},q_{t-1})\}_{t=1}^N$ have the same objective. The objective is determined by $\{p(y_t,x_t,s_{t-1},q_{t-1})\}_{t=1}^N$ since the mutual information at each time is a function of one instance from this set.

This is shown using induction: for $n=1$, write $p(x_1,y_1,s_0,q_0) = p(y_1|x_1,s_0)\tilde{p}(x_1|s_0,q_0)p(s_0,q_0)$, indicating that reward depends on the marginal $\tilde{p}$ only. Assume that $\{p(x_t,y_t,s_{t-1},q_{t-1})\}_{t=1}^N$ is induced by both input distributions and, thus, induce the same $N$th objective. Let us show that $p(x_{N+1},y_{N+1},s_{N},q_{N})$ depends on the marginal $\tilde{p}$ only. First, note that $p(s_{N},q_{N})$ is determined by the $N$th step since $q_{N}$ is a function of $(q_{N-1},y_N)$ and $s_N$ is a function of $(x_N,y_N,s_{N-1})$. Furthermore, $p(x_{N+1}|s_{N},q_{N})$ is identical under both input distributions and $p(y_{N+1}|s_{N},x_{N+1},q_{N})$ is given by the channel specification. Thus, $p(x_{N+1},y_{N+1},s_{N},q_{N})$ is equal under both input distributions.
\end{proof}

\subsection{Proof of Lemma \ref{lemma:step2} (Step $2$):}\label{subsec:step2}
In this section, the goal is to show that stationary inputs are optimal for the upper bound derived in Lemma \ref{lemma:step1}. The first stage is to formulate the upper bound as a DP problem. We then present a known result from the DP literature \cite{DP_discount2avg} that states sufficient conditions for the existence of an optimal stationary policy. Finally, it is proved that these conditions are satisfied in our DP problem and, thus, the existence of an optimal stationary policy is established.

\subsubsection{DP formulation}
The DP definitions presented here follow the formulation in \cite{Arapos93_average_cose_survey}; similar formulations can also be found in \cite{TatikondaMitter_IT09,PermuterCuffVanRoyWeissman08,Sabag_BEC}.

Define the DP state at time $t$ (prior to the $t$th action) as the probability vector $z_{t-1}=P_{S_{t-1},Q_{t-1}}$. As the initial state, $(s_0,q_0)$, lies in a closed communicating class, $A$, the state space is taken as the $|A|$-dimensional unit simplex. Actions are valid conditional distributions $P_{X|S,Q}$ and, specifically, the action at time $t$ is $u_t=P_{X_{t}|S_{t-1},Q_{t-1}}$. The reward gained at time $t$ is taken to be $I(X_t,S_{t-1};Y_t|Q_{t-1})$. Note that this is a deterministic DP as no disturbance is defined.

To show that the above definitions hold for the DP properties, we must verify that there exists a dynamics function and that the reward at time $t$ is a function of $(z_{t-1},u_t)$:

\textbf{Dynamics:} We show that there exists a dynamics function, $F:\mathcal{Z}\times\mathcal{U}\rightarrow \mathcal{Z}$, such that $z_{t}=F(z_{t-1},u_t)$. Each state, $z_{t}$, is a collection of the probabilities $p(s_t,q_t)$, and can be calculated as follows:
\begin{align}\label{eq:dynamics}
  p(s_t,q_t) &= \sum_{y_t,x_t,s_{t-1},q_{t-1}} p(s_t,q_t,y_t,x_t,s_{t-1},q_{t-1}) \nn\\
  &= \sum_{y_t,x_t,s_{t-1},q_{t-1}} p(s_t,q_t|x_t,y_t,s_{t-1},q_{t-1})p(y_t|x_t,s_{t-1})p(x_t|s_{t-1},q_{t-1})p(s_{t-1},q_{t-1})\nn\\
  &\stackrel{(a)}= \sum_{y_t,x_t,s_{t-1},q_{t-1}} \mathbbm{1}_{s_t=f(x_t,y_t,s_{t-1})}\mathbbm{1}_{q_t=g(y_t,q_{t-1})}p(y_t|x_t,s_{t-1})p(x_t|s_{t-1},q_{t-1})p(s_{t-1},q_{t-1}),
\end{align}
where step $(a)$ follows from the facts that the state in a unifilar channel is a function of the triplet $(x_t,y_t,s_{t-1})$, and the stationary context is defined by a function $g:\mathcal{Q}\times \mathcal{Y}\rightarrow\mathcal{Q}$. Recall that $z_{t-1}$ consists of all entries $p(s_{t-1},q_{t-1})$ and the action $u_t= P_{X_t|S_{t-1},Q_{t-1}}$ holds all values of the form $p(x_t|s_{t-1},q_{t-1})$; therefore, each entry in $z_t$ is a time-invariant function of the pair $(z_{t-1},u_t)$.

\textbf{Reward:} Let us show that each reward is a function of the current state and action, i.e., there exists a function $R:\mathcal{Z}\times\mathcal{U}\rightarrow\mathbb{R}$. The reward at time $t$ is $I(X_t,S_{t-1};Y_t|Q_{t-1})$ and is a function of $P_{Y_t,X_t,S_{t-1},Q_{t-1}}$, which can be written as $P_{Y_t|X_t,S_{t-1}}P_{X_t|S_{t-1},Q_{t-1}}P_{S_{t-1},Q_{t-1}}$. The latter factorization of the joint distribution is a function of the state $z_{t-1}=P_{S_{t-1},Q_{t-1}}$, the action $u_t = P_{X_t|S_{t-1},Q_{t-1}}$ and the channel $P_{Y_t|X_t,S_{t-1}}$. From now on, we use the notation $R(z,u)$ for the mutual information that is achieved with a state $z$ and action $u$.

The DP formulation above implies that the optimal average reward is
\begin{align*}
  \rho^{\ast} &=  \sup_{\pi}\liminf _{N\rightarrow\infty} \frac{1}{N}\sum_{t=1}^{N} I(X_t,S_{t-1};Y_{t}|Q_{t-1}),
\end{align*}
where $\pi$ corresponds to a policy, i.e., an infinite sequence of actions. Note that $\rho^\ast$ is equal to the upper bound in \eqref{eq:step1}, so this is an equivalent DP problem for the upper bound calculation.

In addition, we define for $\beta<1$ and initial state $\xi\in\mathcal{Z}$ their optimal discounted reward as
\begin{align*}
  \nu_\beta(\xi) &=  \sup_{\pi}\sum_{t=1}^{\infty}\beta^t I(X_t,S_{t-1};Y_{t}|Q_{t-1}).
\end{align*}

\subsubsection{Sufficient conditions for the existence of an optimal stationary policy}
\begin{itemize}
  \item[\textbf{(C1)}] The transition kernel is continuous with respect to weak convergence in $P(\mathcal{Z})$. In our case, the transition kernel is defined by the dynamics function, $F(\cdot,\cdot)$, in \eqref{eq:dynamics}.
  \item[\textbf{(C2)}] The state space, $\mathcal{Z}$, is a locally compact space with a countable base.
  \item[\textbf{(C3)}] The multifunction $\mathcal{U}(z)$ is upper semi-continuous. The notation $\mathcal{U}(z)$ stands for allowed actions at state $z$. In our case, $\mathcal{U}(z)$ is the set of all conditional distributions of the form $P_{X|S,Q}$, i.e., the set of allowed actions equals the full set of actions for all $z$.
  \item[\textbf{(C4)}] The reward function $R(z,u)$ is lower semi-continuous in $(z,u)$.
  \item[\textbf{(C5)}] Let $m_\beta=\sup_z \nu_\beta(z)$; then $\sup_{\beta<1} \{m_\beta-\nu_\beta(z)\}<\infty$ for all $z\in\mathcal{Z}$.
\end{itemize}
In \cite{DP_discount2avg}, the above conditions were presented for a model where the optimal average reward is defined as the minimization over all policies. Since our model is defined as a maximization problem, trivial modifications should be made in \textbf{(C3)}-\textbf{(C4)}; however, we will show that in our problem these conditions are satisfied in both the upper and lower cases.

\begin{theorem}[Theorem 3.8, \cite{DP_discount2avg}]\label{lemma:stationary}
If \textbf{(C1)}-\textbf{(C5)} are satisfied then there exists an optimal stationary policy for the average reward DP problem.
\end{theorem}

Returning to our problem, we will show that \textbf{(C1)}-\textbf{(C5)} are satisfied and this leads to the conclusion that there exists an optimal stationary policy.
\begin{proof}[Conditions \textbf{(C1)}-\textbf{(C5)} are satisfied in our problem]

 \textbf{(C1)} The transition kernel is continuous with respect to weak convergence if the following is satisfied: for all $v(\cdot)\in\mathcal{C}_b(\mathcal{Z})$(continuous and bounded functions on $\mathcal{Z}$),
 \begin{equation}\label{eq:weak_convergence}
\int_\mathcal{Z} v(y)F(dy|\cdot,\cdot)\in\mathcal{C}_b(\mathcal{Z}\times\mathcal{U}).
\end{equation}
The transition kernel, $F(dy|z,u)$, is a dirac measure and, therefore, integration over $y$ \eqref{eq:weak_convergence} returns $v(F(z,u))$.

The function $v(F(z,u))$ is bounded since $v(\cdot)$ is bounded. For the continuity, note that by \eqref{eq:dynamics} each element in $F(z,u)$ is continuous with respect to $(z,u)$ (in any norm) since it is a finite sum of elements in $(z,u)$. Since the function $v(\cdot)$ is continuous, the composition of $F(z,u)$ into $v(\cdot)$ is continuous in $(z,u)$. To conclude, the composition $v(F(z,u))$ is bounded and continuous with respect to $(z,u)$.

\noindent\textbf{(C2)} The state space is the $|A|$-dimensional unit simplex. As the n-dimensional simplex is a closed subset of the n-cubic with a unit edge, it is locally compact with a countable base.

\noindent\textbf{(C3)} The general scenario is where the action space can depend on $z$; however, in our problem $\mathcal{U}(z)$ is constant in $z$ and, thus, trivially continuous in $z$.

\noindent\textbf{(C4)} The mutual information can be written as a sum of entropies, where each entropy is continuous in the joint distribution of $(y,x,s,q)$ that is induced by $(z,u)$; therefore, it is both lower and upper semi-continuous.

\noindent\textbf{(C5)} By \cite[Proposition 2.1]{DP_discount2avg}, conditions $(C1)-(C4)$ imply that there exists an optimal stationary policy for the discounted problem, which is denoted here as $f_\beta  = P_\beta^\ast(x|s,q)$. The policy $f_\beta$ implies a structure on the $(S,Q)$-graph and might result in several closed communicating classes in the case where there are edges with probability zero. Denote by $\mathcal{A}$ the $(S,Q)$-graph after removing edges with $P_\beta^\ast(x|s,q)=0$. It is convenient to partition the analysis for two cases based on the structure of $\mathcal{A}$:
\begin{itemize}
  \item Case A: The graph induced by the policy $f_\beta$, $\mathcal{A}$, has a single closed communicating class.
  \item Case B: The graph induced by the policy $f_\beta$, $\mathcal{A}$, has more than one closed communicating class.
\end{itemize}

\textbf{Case A}: With some abuse of terminology, we will refer to $\mathcal{A}$ as the closed communicating class in the $(S,Q)$-graph, since all nodes outside this class are inessential in the infinite-horizon regime. Denote by $T$ the transfer matrix induced by $f_\beta$ on the single closed class, and by $D$ its period. Since $\mathcal{A}$ is irreducible, the graph can be partitioned into $A_0,A_1,\dots A_{D-1}$ disjoint classes on a cycle based on a period equivalence. The stationary distribution of the Markov chain on $\mathcal{A}$ is denoted by $\pi_{f_\beta}$.

Consider the $D$-blocks Markov chain and, specifically, a Markov chain with transition matrix $T^D$. Since $D$ is the period of the original graph, the new Markov chain implies $D$ separate aperiodic and irreducible Markov chains. Each Markov chain is on a class $A_d$ and we denote by $\pi(A_d)$ the stationary distribution of each class $A_d$ where $d\in[0:D-1]$.

For initial state $\xi\in\mathcal{P}(Z)$, denote its weights vector as $W(\xi)$ with $d$ inputs, where the $d$th input is $w_d(\xi) = \sum_{(s,q)\in A_d}\xi(s,q)$. Define for all $k$:
\begin{align}\label{eq:D_limit}
  \pi_k(\xi)&\triangleq [w_{[k]}(\xi)\pi(A_0),w_{[k+1]}(\xi)\pi(A_1),\dots,w_{[k+D-1]}(\xi)\pi(A_{D-1})],
\end{align}
where the indices with $[\cdot]$ are taken modulo $D$. Finally, the vectors in \eqref{eq:D_limit} are used to define
\begin{align*}
  \nu^\pi_\beta(\xi)&= \sum_{n=1}^\infty \beta^n R(\pi_{n}(\xi),f_\beta)\\
  \nu^\ast_\beta &= \sum_{n=1}^\infty \beta^n R(\pi_{f_\beta},f_\beta).
\end{align*}

The expression $\nu^\pi_\beta(\xi)$ corresponds to the discounted reward that is achieved with states that are moved periodically through all possibilities in \eqref{eq:D_limit}. The first step is to show that for a fixed initial state, the actual reward and its corresponding periodic reward, $\nu^\pi_\beta(\xi)$ are bounded as follows:
\begin{lemma}\label{lemma:beta_upper_bound}
For all initial states, $\xi$,
\begin{align*}
\sup_\beta |\nu_\beta(\xi) -  \nu^\pi_\beta(\xi)| <\infty.
\end{align*}
\end{lemma}

The second step of the proof is to show that the achieved periodic reward (which is a function of the initial state) is bounded with respect to some constant quantity, which does not depend on $\xi$:
\begin{lemma}\label{lemma:beta_upper_ast}
For all initial states, $\xi$,
\begin{align*}
\sup_\beta |\nu^\ast_\beta -  \nu^\pi_\beta(\xi)| <\infty.
\end{align*}
\end{lemma}
A direct conclusion from the above two lemmas is the required condition \textbf{(C5)}:
\begin{align*}
  \sup_\beta |\nu_\beta(\xi) -  \nu_\beta(\xi')| &\stackrel{(a)}\leq 2 \sup_\beta \max_\xi|\nu_\beta(\xi) -  \nu^\ast_\beta| \\
  &\stackrel{(a)}\leq 2 \sup_\beta \max_\xi|\nu_\beta(\xi) - \nu^\pi_\beta(\xi)| + |\nu^\pi_\beta(\xi) - \nu^\ast_\beta|\\
  &\stackrel{(b)}< \infty,
\end{align*}
where $(a)$ follows from the triangle's inequality and $(b)$ follows from Lemma \ref{lemma:beta_upper_bound} and Lemma \ref{lemma:beta_upper_ast}.

The proofs of Lemma \ref{lemma:beta_upper_bound} and Lemma \ref{lemma:beta_upper_ast} appear in Appendix \ref{app:lemma_beta_1} and Appendix \ref{app:lemma_beta_2}, respectively.

Their proof requires the following preliminaries on total variation distance and Markov chains.
\begin{definition}
For finite alphabet, $\mathcal{X}$, and two probability mass functions, $P_X$ and $Q_X$, the total variation distance is
\begin{align*}
  ||P_X-Q_X||_{TV}&\triangleq \frac{1}{2}\sum_x |p(x)-q(x)|.
\end{align*}
\end{definition}
The following Lemma summarizes two properties of the total variation distance:
\begin{lemma}[Lemma V.I-V.II, \cite{Paul_synthesis}]\label{lemma:Paul}
For two joint PMFs, $P$ and $Q$ on $\mathcal{X}\times\mathcal{Y}$, their total variance satisfies
\begin{align*}
\| P_X - Q_X\|_{TV}&\leq \| P_{X,Y} - Q_{X,Y}\|_{TV},
\end{align*}
and the equality holds if $P_{Y|X}=Q_{Y|X}$.
\end{lemma}

The following is an upper bound on the convergence rate of aperiodic Markov chains.
\begin{lemma}[Theorem 4.8,\cite{Levin_Markovs}]\label{lemma:convergence}
Let $T$ be a transfer matrix of an irreducible and aperiodic Markov chain on a space $\mathcal{X}$ with a stationary distribution $\pi$. Then there exist
constants $\alpha \in (0,1)$ and $C > 0$ such that
\begin{equation*}
  \max_{\xi\in\mathcal{P}(X)} \| \xi T^n - \pi\|_{TV}\leq  C\alpha^n.
\end{equation*}
\end{lemma}

The Markov chain in our problem is not necessarily aperiodic; therefore, a slight adaptation of Lemma \ref{lemma:convergence} for the periodic case is now given.
\begin{lemma}[Convergence of Periodic Markov Chains]\label{lemma:convergence_periodic}
Let $T$ be a transfer matrix of an irreducible Markov chain with period $D$ on a space $\mathcal{X}$. Then there exist
constants $\alpha \in (0,1)$ and $C > 0$ such that for all $\xi$
\begin{equation*}
  \| \xi T^{nD+k} - \pi_k(\xi)\|_{TV}\leq  C\alpha^n,
\end{equation*}
where $\pi_k(\xi)$ are defined in \eqref{eq:D_limit}.
\end{lemma}

\begin{proof}[Proof of Lemma \ref{lemma:convergence_periodic}]
For some $k\in[0:D-1]$ and for all $\xi$, consider
\begin{align*}
  || \xi T^{nD+k} - \pi_k(\xi)||_{TV} &= \frac{1}{2}\sum_{(s,q)} |\xi P^{nD+k}(s,q) - \pi_k(\xi)(s,q) |\\
  &=\sum_d \frac{1}{2}\sum_{(s,q)\in A_d} |\xi T^{nD+k}(s,q) - \pi_k(\xi)(s,q) |\\
  &\stackrel{(a)}=\sum_d \frac{1}{2}\sum_{(s,q)\in A_d} |\xi T^{nD+k}(s,q) - w_{[k+d]}(\xi)\pi(A_d)(s,q) |\\
  &=\sum_d w_{[k+d]}(\xi) \frac{1}{2}\sum_{(s,q)\in A_d} \left|\frac{\xi T^{nD+k}(s,q)}{w_{[k+d]}(\xi)} - \pi(A_d)(s,q) \right|\\
  &\stackrel{(b)}=\sum_d w_{[k+d]}(\xi) \left|\left| \frac{\xi T^{nD+k}}{w_{[k+d]}(\xi)} - \pi(A_d)\right|\right|_{TV} \\
  &\stackrel{(c)}\leq \sum_d w_{[k+d]}(\xi) C_d \alpha_d^n\\
  &\leq \sum_d w_{[k+d]}(\xi) \max_d C_d \alpha_d^n\\
  &\stackrel{(d)}\triangleq  C \alpha^n,
\end{align*}
where $(a)$ follows by substituting Eq. \eqref{eq:D_limit}, $(b)$ follows by the total variation distance definition when conditioned on the class $A_d$, $(c)$ follows from Lemma \ref{lemma:convergence} and $(d)$ follows by $\sum_d w_{[k+d]}(\xi)=1$.
\end{proof}

\textbf{Case B:} We give an outline of the proof for case B, as it essentially follows the same steps used for Case A. In this scenario, there are several closed communicating classes, denoted by $A_1,\dots,A_k$, and their corresponding periods are $D_1,\dots,D_k$. The technique used in Case A is composed of two steps: the first is to show that the reward is bounded with a reward that has some periodic behavior, as argued in Lemma \ref{lemma:beta_upper_bound}, and the second step is to show that this periodic reward is bounded with respect to some constant quantity (with respect to the initial state).

The first step is addressed by studying the periodic behavior of each closed class, as was done in Case A. Clearly, the initial weight of each closed class is time-invariant since weight cannot move between closed classes. It follows that the common period of all classes is simply the multiplication of all periods, i.e., $D=\prod_i D_i$. This concludes the analysis that is required for the first part of the proof.
The second part of the proof follows the lines used for the proof of Lemma \ref{lemma:beta_upper_ast}; specifically, the upper bound derivation can be followed with the defined $D$, and the $\epsilon$-policy construction is identical.
\end{proof}
\subsection{Proof of Lemma \ref{lemma:step3} (Step $3$)}\label{subsec:step3}
Before presenting the proof of Lemma \ref{lemma:step3}, we impose a restriction on the maximization domain:
\begin{lemma}\label{lemma:pi}
It is sufficient to take the supremum in \eqref{eq:step3} over $P_{X|S,Q}$ which lies in $\mathcal{P}_\pi$.
\end{lemma}

\begin{proof}[Proof of Lemma \ref{lemma:pi}]
In this proof, we will take the maximizer of the LHS in \eqref{eq:step3}, and show that there exists a distribution from $\mathcal{P}_\pi$ that induces the maximal reward.

Let $P_{X|S,Q}^\ast$ be a maximizer which implies two closed communicating classes, $A_1$ and $A_2$, with average rewards, $R_1$ and $R_2$, respectively. Construct $\tilde{P}_{X|S,Q}$ exactly as $P_{X|S,Q}^\ast$, but where positive probabilities are given for edges from $A_1$ to $A_2$. This modification is legitimate since the $(S,Q)$-graph is irreducible. This modification did not effect the rewards for initial states in $A_2$, while the rewards of initial states in $A_1$ might be changed.

By the optimality of $P_{X|S,Q}^\ast$, the reward $R_1$ cannot be increased and, therefore, $R_1=R_2$. Since $R_1=R_2$, the policy induced by $\tilde{P}_{X|S,Q}\in\mathcal{P}_\pi$ achieves the same optimal rewards as the maximizer. The above argument can be extended to any number of closed communicating classes since the graph is finite.
\end{proof}

By the definition of the set $\mathcal{P}_\pi$, there is a single closed communicating class for each input distribution $P_{X|S,Q}$. Let $A$ denote the graph which describes the closed class and let $T_A$ be its transition probability matrix. Since $A$ is irreducible, there exists a stationary distribution $\pi = [\pi_1,\pi_2,\dots,\pi_{|A|}]$ which is the unique solution for the equation $\pi T_A=\pi$. Here, the stationary distribution is in the Cesaro sum sense, i.e.,
\begin{equation*}
\frac{1}{n}\sum_{m=1}^n T_A^m\rightarrow \left(
                               \begin{array}{cccc}
                                 \pi_1 & \pi_2 & \dots & \pi_{|A|} \\
                                 \vdots &  &  & \vdots \\
                                 \pi_1 & \pi_2 & \dots & \pi_{|A|} \\
                               \end{array}
                             \right).
\end{equation*}

Let $|\mathcal{D}|$ be the period of the graph $A$, and let $A_1,\dots A_{|\mathcal{D}|}$ be the disjoint subsets of nodes based on a period equivalence. The dependence of $A$, $|\mathcal{D}|$ and $A_i$ on $P_{X|S,Q}$ is omitted.

\begin{proof}[Proof of Lemma \ref{lemma:step3}]
Consider the following chain of equalities:
\begin{align}\label{eq:single_letter}
&\sup_{P_{X|S,Q}} \liminf_{N\rightarrow \infty}\frac{1}{N} \sum_{i=1}^N I(X_i,S_{i-1};Y_i|Q_{i-1}) \nn\\
&\stackrel{(a)}=\sup_{P_{X|S,Q}\in\mathcal{P}_\pi} \liminf_{N\rightarrow \infty}\frac{1}{N} \sum_{i=1}^N I(X_i,S_{i-1};Y_i|Q_{i-1}) \nn\\
&\stackrel{(b)}\leq\sup_{P_{X|S,Q}\in\mathcal{P}_\pi} \liminf_{N\rightarrow \infty}\frac{1}{N|\mathcal{D}|} \sum_{i=1}^{N|\mathcal{D}|} I(X_i,S_{i-1};Y_i|Q_{i-1}) \nn\\
&\stackrel{(c)}= \sup_{P_{X|S,Q}\in\mathcal{P}_\pi} \liminf_{N\rightarrow \infty}\frac{1}{N|\mathcal{D}|} \sum_{d=1}^{|\mathcal{D}|} \sum_{i=0}^{N-1} I(X_{i|\mathcal{D}|+d},S_{i|\mathcal{D}|+d-1};Y_{i|\mathcal{D}|+d}|Q_{i|\mathcal{D}|+d-1})\nn\\
&\stackrel{(d)}= \sup_{P_{X|S,Q}\in\mathcal{P}_\pi} \frac{1}{|\mathcal{D}|}\sum_{d=1}^{|\mathcal{D}|} \liminf_{N\rightarrow \infty}\frac{1}{N}  \sum_{i=0}^{N-1} I(X_{i|\mathcal{D}|+d},S_{i|\mathcal{D}|+d-1};Y_{i|\mathcal{D}|+d}|Q_{i|\mathcal{D}|+d-1})\nn\\
&\stackrel{(e)}= \sup_{P_{X|S,Q}\in\mathcal{P}_\pi} \frac{1}{|\mathcal{D}|}\sum_{d\in\mathcal{D}}  I(X,S_{d};Y|Q_{d})\nn\\
&\stackrel{(f)}= \sup_{P_{X|S,Q}\in\mathcal{P}_\pi} I(X,S;Y|Q,D)\nn\\
&= \sup_{P_{X|S,Q}} H(Y|Q,D) - H(Y|X,S,Q,D)\nn\\
&\stackrel{(g)}\leq \sup_{P_{X|S,Q}\in\mathcal{P}_\pi} H(Y|Q) - H(Y|X,S,Q)\nn\\
&= \sup_{P_{X|S,Q}\in\mathcal{P}_\pi} I(Y;X,S|Q),
\end{align}
where
\begin{itemize}
  \item[(a)] follows from Lemma \ref{lemma:pi};
  \item[(b)] follows by taking the limit on a subsequence of $N$, i.e., the sequence $|\mathcal{D}|,2|\mathcal{D}|,\dots$;
  \item[(c)] follows by re-indexing the summation in blocks of $\mathcal{D}$ elements;
  \item[(d)] follows by exchanging the order of the limit and sum due to the limit existence of each term in the sum;
  \item[(e)] follows by calculating the limit for a fixed $d$. Specifically, the value of $d$ determines a class $A_d$. The distribution of $p(s_{i|\mathcal{D}|+d-1},q_{i|\mathcal{D}|+d-1})$ tends to the stationary distribution since the chain is aperiodic and irreducible. This, in turn, gives that the distribution for each $d$ is $p_d(y,x,s,q)=p(y,x|s,q)\pi_d(s,q)$, where
       \begin{align}\label{eq:old_distr}
  \pi_d(s,q)&= \begin{cases} \frac{\pi(s,q)}{\sum_{(s,q)\in A_d} \pi(s,q)} &\mbox{if } (s,q)\in A_d; \\
0 & \mbox{otherwise.} \end{cases}.
\end{align}
\item[(f)] follows by defining a uniform RV, $D$, on $[1:|\mathcal{D}|]$. The joint distribution is $p(y,x,s,q,d) = p(y,x|s,q)\pi_d(s,q)p(d)$;
\item[(g)] follows from the Markov $Y-(X,S)- D$ and the fact that conditioning reduces entropy. This expression is calculated with respect to $p(s,q,x,y)$, which is the marginal distribution of $p(d)\pi_d(s,q)p(x|s,q)p(y|x,s)$. Explicit calculation gives that
\begin{align*}
  p(s,q) &= \sum_d p(d)\pi_d(s,q)\\
         &\stackrel{(\ast)}= \frac{1}{|\mathcal{D}|} \frac{\pi(s,q)}{\sum_{(s,q)\in A_d} \pi(s,q)}\\
         &\stackrel{(\ast\ast)}= \frac{1}{|\mathcal{D}|} \frac{\pi(s,q)}{\frac{1}{\mathcal{|D|}}}\\
         &= \pi(s,q),
\end{align*}
where $(\ast)$ is obtained by substituting the expression from \eqref{eq:old_distr} and $(\ast\ast)$ follows from $\sum_{(s,q)\in A_d} \pi(s,q) = \frac{1}{|\mathcal{D}|}$, for all $d$, since each class is on a cycle.
\end{itemize}
To conclude the proof, we have shown in \eqref{eq:single_letter} that $I(X,S;Y|Q)$ with $\pi(s,q)p(x|s,q)p(y|x,s)$ is an upper bound.
\end{proof}
\section{Conclusions}\label{sec:conclusions}
An upper bound on the feedback capacity of unifilar state channels was derived. The upper bound is expressed by a computable single-letter expression and it was shown how the bound can be computed for known capacity results. Calculation of the upper bound for the DEC resulted a new capacity result together with the sufficient condition for the optimality of the upper bound. For all studied channels, the optimal $Q$-graph was obtained from DP simulations. A further direction that is under investigation is a structured method for finding such an optimal $Q$-graph without DP simulations.

The upper bound gives a useful insight into the structure of optimal output processes. Specifically, as the bound is tight for all known capacities, this provides a unifying structure for the optimal output processes. The technique used in this paper might also be applied to any entropy rate of a random process. Specifically, for the $n$th instance, $H(Y_n|Y^{n-1})$, the process history can be quantized using a $Q$-graph. However, even for FSCs without feedback, the obtained upper bound is not computable since the contexts are not revealed to the encoder.

\appendices
\section{Proof of Lemma \ref{lemma:coupled}}\label{app:proof_lemma_coupled}
Each node has an outgoing edge since for each $s$ there exists $(x,y)$ such that $p(y|x,s)>0$. Therefore, each node $(s,q)$ has an outgoing edge $(s,q)\rightarrow(g(q,y), f(s,x,y))$. It should be clear that each node has at least one outgoing edge to another node since, if a node $(s,q)$ has edges to itself only, this means that for all $(x,y)$, $s=f(x,y,s)$, implying $|\mathcal{S}|=1$. Therefore, by the pigeonhole principle, there exists at least one closed communicating class.

To show that each closed class has all $q\in Q$, recall that the $Q$-graph is irreducible and, therefore, for each pair $(q_0,q_n)$, there exists a path $q_0\rightarrow q_n$ labelled by $y_1\dots y_n$ such that $q_i=g(y_i,q_{i-1})$. For the first label, $y_1$, there exists $(x_1,s_1)$ such that $p(y_1|x_1,s_1)>0$. Then, for a node $(s,q_1)$ in the closed class there is an edge to $(f(y_1,x_1,s_1),q_2)$; this argument can be repeated until a node of the form $(\cdot,q_n)$ is reached. Since it is a closed communicating class, each path leads to a node in this class.

The proof that each closed class has all $s\in S$ is similar to the previous argument, but using the strongly connected property of states, i.e., that the states graph is irreducible. For each $s,s'$, there exists a path labelled by $(x_1,y_1),(x_2,y_2),\dots$ with probabilities $p(y_i|x_i,s_{i-1})>0$ such that $s$ reaches $s'$. Therefore, for each $(s,q)$ there is a path to $(s',\cdot)$ for all $s'$. $\hfill\blacksquare$

\section{Proof of Theorem \ref{theorem:lower}}
In this proof, we show that BCJR-invariant inputs induce the Markov chain $Y_t-Q_{t-1}-Y^{t-1}$ for all $t$. This Markov chain gives, in turn, that the feedback capacity expression with the chosen input is $I(X,S;Y|Q)$.

Since inputs $P_{X|S,Q}$ are assumed to be aperiodic inputs, it may be assumed that the initial distribution is $\pi(s,q)$ since it is reached with high probability. We will show by induction that the value of the BCJR estimator is determined by a context of sequence, i.e., $p(S_{t}|y^{t}) = \pi_{S|Q=q}$, where $q$ is the context of $y^t$. At time $t-1$, assume that $p(S_{t-1}|y^{t-1}) = \pi_{S|Q=q}$ where $q$ is the context of $y^{t-1}$. Then, one can calculate at time $t$,
\begin{align*}
  p(S_{t}|y^{t}) &\stackrel{(a)}= B(\pi_{S|Q=q},y_t)\\
  &\stackrel{(b)}= \pi_{S|Q=g(q,y_t)},
\end{align*}
where $(a)$ follows from the induction hypothesis and the forward-recursive relation, \eqref{eq:BCJR}, and $(b)$ follows from the  BCJR-invariant property. Thus, we have shown that the probability vector $p(S_{t}|y^{t})$ is determined by the context of the sequence, which concludes the proof.

Given the fact that $p(S_{t-1}=s|y^{t-1}) = \pi_{S=s|Q=q}$, we can show that the Markov chain $Y_t-Q_{t-1}-Y^{t-1}$ holds:
\begin{align}\label{eq:Lower_Markov}
  p(y_t|y^{t-1},q_{t-1}) &= \sum_{s_{t-1},x_t}   p(y_t,x_t,s_{t-1}|y^{t-1},q_{t-1})\nn\\
         &= \sum_{s_{t-1},x_t}   p(y_t|x_t,s_{t-1})p(x_t|s_{t-1},q_{t-1},y^{t-1})p(s_{t-1}|y^{t-1},q_{t-1})\nn\\
         &\stackrel{(a)}= \sum_{s_{t-1},x_t}   p(y_t|x_t,s_{t-1})p(x_t|s_{t-1},q_{t-1})\pi(s_{t-1}|q_{t-1})\nn\\
         &= p(y_t|q_{t-1}),
\end{align}
where $(a)$ follows from the fact that the input $x_t$ depends on the pair $(s_{t-1},q_{t-1})$ only, and the from the above inductive argument which shows that $p(s_{t-1}|y^{t-1},q_{t-1})=\pi(s_{t-1}|q_{t-1})$.

Finally, the theorem can be proved by the following chain of inequalities
\begin{align*}
C_{\text{fb}} &\stackrel{(a)}= \sup_{\{p(x_t|s_{t-1},y^{t-1})\}_{t\geq 1}} \liminf_{N\rightarrow \infty} \frac{1}{N} \sum_{i=1}^N   I(X_i,S_{i-1};Y_i|Y^{i-1})\\
              &\stackrel{(b)}= \sup_{\{p(x_t|s_{t-1},y^{t-1})\}_{t\geq 1}} \liminf_{N\rightarrow \infty} \frac{1}{N} \sum_{i=1}^N    I(X_i,S_{i-1};Y_i|Q_{i-1}) - I(Y_i;Y^{i-1}|Q_{i-1})\\
              &\stackrel{(c)}\geq \liminf_{N\rightarrow \infty} \frac{1}{N} \sum_{i=1}^N   I(X_i,S_{i-1};Y_i|Q_{i-1}) - I(Y_i;Y^{i-1}|Q_{i-1})\\
              &\stackrel{(d)}= \liminf_{N\rightarrow \infty} \frac{1}{N} \sum_{i=1}^N   I(X_i,S_{i-1};Y_i|Q_{i-1}) \\
              &\stackrel{(e)}= I(X,S;Y|Q),
\end{align*}
where
 \begin{itemize}
   \item[(a)] follows from the capacity formula from Theorem \ref{theorem:capacity_unifilar};
   \item[(b)] follows from re-writing $I(X_i,S_{i-1};Y_i|Y^{i-1})=H(Y_i|Y^{i-1})-H(Y_i|X_i,S_{i-1})$ and adding $H(Y_i|Q_{i-1})-H(Y_i|Q_{i-1})$;
   \item[(c)] follows by taking the input distribution to be $p(x_t|s_{t-1},y^{t-1})=P_{X|S,Q}$ for all $t$;
   \item[(d)] follows from the Markov chain $Y_i-Q_{i-1}-Y^{i-1}$ in \eqref{eq:Lower_Markov};
   \item[(e)] follows from the aperiodic Markov chain on the state space $(S,Q)$ which induces its corresponding stationary distribution.
 \end{itemize}
$\hfill\blacksquare$
\section{Proof of Corollary \ref{coro:trapdoor_05}}\label{app:coro_trapdoor_05}
In this section we show that $C_{Trap}(0.5)\leq\log_2 \phi$.
The upper bound on the capacity of the trapdoor channel with $p=0.5$ from Theorem \ref{theorem:trapdoor_Q1} is:
\begin{align*}
C_{\text{Trap}}(0.5)&\leq \max_{(\alpha_1,\alpha_2,\alpha_3)\in[0,1]^3}  \lambda_1(\alpha_1,\alpha_2,\alpha_3) + \lambda_2(\alpha_1,\alpha_2,\alpha_3),
\end{align*}
where
\begin{align*}
  \lambda_1(\alpha_1,\alpha_2,\alpha_3)&= 2(\kappa_1+\kappa_2)H_2\left(\frac{\kappa_1(1-0.5\alpha_1) + 0.5\kappa_2\alpha_2}{\kappa_1+\kappa_2}\right)\nn\\
  \lambda_2(\alpha_1,\alpha_2,\alpha_3)&= 2(\kappa_3 - \kappa_1\alpha_1 - \kappa_2\alpha_2 - 0.5\alpha_3)
\end{align*}
and
\begin{align*}
  \delta &= 3\alpha_1 -\alpha_2 + \alpha_1\alpha_3 - \alpha_1\alpha_2 + \alpha_2\alpha_3 - 2\alpha_3 + 2 \\
  \kappa_1 &= \frac{(1-\alpha_3)(1 - 0.5 \alpha_2)}{\delta} \\
  \kappa_2 &= \frac{0.5\alpha_1(1 + \alpha_3)}{\delta} \\
  \kappa_3 &= \frac{\alpha_1(1 -  0.5\alpha_2)}{\delta}.
\end{align*}

The proof will follow from the facts that $\lambda_2(\cdot)\leq0$ and $\lambda_1(\cdot)\leq\log\phi$. Let us begin with $\lambda_2(\cdot)\leq0$ that is equal to
\begin{align}\label{eq:trapdoor_g_explicit}
\lambda_2(\cdot)&= \frac{\alpha_2\alpha_3 - \alpha_1\alpha_2 - \alpha_1\alpha_3 - 2\alpha_3 - \alpha_1\alpha_3^2 - \alpha_2\alpha_3^2 + 2\alpha_3^2 - \alpha_1\alpha_2\alpha_3}{\delta}.
\end{align}
Since $\delta>0$, it is sufficient to verify that the numerator is always negative; to this end, we write the numerator of \eqref{eq:trapdoor_g_explicit} as a polynomial of $\alpha_3$ when $\alpha_1,\alpha_2$ are some parameters:
\begin{align*}
&\alpha_3^2 (2-\alpha_1-\alpha_2)+ \alpha_3 (-2-\alpha_1+\alpha_2 - \alpha_1\alpha_2) - \alpha_1\alpha_2.
\end{align*}
The coefficient of $\alpha_3^2$ is positive and, therefore, is a convex function. It can also be noted that the function is negative both at $\alpha_3=0$ and $\alpha_3=1$, thus, for $\alpha_3\in[0,1]$ it is upper bounded with zero.

Let us provide an upper bound for $\lambda_1(\alpha_1,\alpha_2,\alpha_3)$:
 \begin{align*}
   \lambda_1(\alpha_1,\alpha_2,\alpha_3)&= 2(\kappa_1+\kappa_2)H_2\left(\frac{\kappa_1(1-0.5 \alpha_1) + \kappa_2 0.5\alpha_2}{\kappa_1+\kappa_2}\right)\\
   &\stackrel{(a)}= 2(\kappa_1+\kappa_2)H_2\left(\frac{\kappa_10.5\alpha_1 + \kappa_2(1-0.5\alpha_2)}{\kappa_1+\kappa_2}\right)\\
   &\stackrel{(b)}= 2(\kappa_1+\kappa_2) H_2\left(\frac{\kappa_3}{\kappa_1+\kappa_2}\right)\\
   &\stackrel{(c)}= 2 (\kappa_1+\kappa_2)H_2\left(\frac{0.5-(\kappa_1+\kappa_2)}{\kappa_1+\kappa_2}\right)\\
   &\stackrel{(d)}= 2 \frac{1}{2(p+1)}H_2\left(p\right)\\
  &\stackrel{(e)}\leq \max_{0\leq p\leq 1}\frac{H_2(p)}{1+p}\\
  &= \log_2\phi,
 \end{align*}
 where
 \begin{itemize}
   \item[(a)] follows from the symmetry of the entropy function, i.e., $H_2(p)=H_2(1-p)$;
   \item[(b)] follows from $\kappa_10.5\alpha_1 + \kappa_2(1-0.5\alpha_2) = \kappa_3$;
   \item[(c)] follows from $\kappa_1+\kappa_2+\kappa_3=0.5$;
   \item[(d)] follows by defining a new variable $p(\alpha_1,\alpha_2,\alpha_3)=\frac{0.5-(\kappa_1+\kappa_2)}{\kappa_1+\kappa_2}$;
   \item[(e)] follows by taking the maximum over $p$, which is obviously restricted to $[0,1]$.
 \end{itemize}

Finally, we can show that
 \begin{align*}
C_{Trap}(p)&\leq \max_{(\alpha_1,\alpha_2,\alpha_3)\in[0,1]^3}  \lambda_1(\alpha_1,\alpha_2,\alpha_3) + \lambda_1(\alpha_1,\alpha_2,\alpha_3)\\
&\leq \max_{(\alpha_1,\alpha_2,\alpha_3)\in[0,1]^3}  \lambda_1(\alpha_1,\alpha_2,\alpha_3) + \max_{(\alpha_1,\alpha_2,\alpha_3)\in[0,1]^3}  \lambda_2(\alpha_1,\alpha_2,\alpha_3)\\
&\stackrel{(a)}\leq \log_2\phi,
\end{align*}
where $(a)$ follows from the derived upper bounds on each function separately. $\hfill\blacksquare$
\section{Proof of Theorem \ref{theorem:trapdoor_lower}}\label{app:trapdoor_lower}
The proof is based on Theorem \ref{theorem:lower} with the $Q$-graph from Fig. \ref{fig:Q2_trapdoor} and the following input distribution:
\begin{align*}
  p(x=0|s=0,q_1) &= 1 \\
  p(x=0|s=0,q_2) & =1 \\
  p(x=0|s=0,q_3) &=  \frac{z p}{1-(1-p)z}\\
  p(x=0|s=0,q_4) &=  \frac{z p}{1-(1-p)z}\\
  p(x=1|s=1,q_1) &=  \frac{z p}{1-(1-p)z}\\
  p(x=1|s=1,q_2) &=  \frac{z p}{1-(1-p)z}\\
  p(x=1|s=1,q_3) &= 1 \\
  p(x=1|s=1,q_4) & =1,
\end{align*}
where $z$ is a parameter in $[0,1]$ and $p$ is the channel parameter. Straightforward calculation gives that $[\pi(s=0|q_1),\pi(s=0|q_2),\pi(s=0|q_3),\pi(s=0|q_4)] = [(1-p)z, 1-z, z, 1-(1-p)z]$.

The BCJR equation can be written as:
\begin{equation*}
p(s_i=0|q_j=g(q_i,y))=\left\{\begin{array}{cc}
 \frac{\delta_i}{(1-p)(\delta_i-\gamma_i) + p\pi(s=0|q_i) + (1-p)\pi(s=1|q_i)} & \text{if } y_t=0, \\
 \frac{(1-p)(\pi(s=0|q_i)-\delta_i) + p(\pi(s=1|q_i) -\gamma_i)}{(1-p)(\pi(s=0|q_i)-\delta_i) + \pi(s=1|q_i) + (1-p)(\gamma_i - \pi(s=1|q_i))} & \text{if } y_t=1, \end{array}\right.
\end{equation*}
where $\delta_i = \pi(s=0|q_i)p(x=0|s=0,q_i)$ and $\gamma_i= \pi(s=1|q_i)p(x=1|s=1,q_i)$. The explicit calculation of the BCJR-invariant property is omitted here as it is identical to the calculations for the DEC and the input-constrained BEC.

For simplicity, we denote $\alpha\triangleq\frac{z p}{1-(1-p)z}$ which can take any value on $[0,1]$, and then we have the stationary vector of the $Q$-graph:
\begin{align*}
  [\pi(q_1),\pi(q_2),\pi(q_3),\pi(q_4)]&=\left[\frac{1-\alpha}{4-2\alpha},\frac{1}{4-2\alpha},\frac{1}{4-2\alpha},\frac{1-\alpha}{4-2\alpha}\right],
\end{align*}
and the per node rewards:
\begin{align*}
  I(X,S;Y|Q=q_1)&=I(X,S;Y|Q=q_4)= z H_2(p) \\
  I(X,S;Y|Q=q_2)&=I(X,S;Y|Q=q_3)= H_2(\alpha) - (1-\alpha)z H_2(p).
\end{align*}

Then, the lower bound can be computed:
\begin{align*}
  C_{\text{fb}}&\geq  I(X,S;Y|Q) \nn\\
               &= 2 \cdot \frac{(1-\alpha)z H_2(p)}{4-2\alpha} + 2 \cdot \frac{H_2(\alpha) - (1-\alpha)z H_2(p)}{4-2\alpha}\nn\\
               &= \frac{H_2(\alpha)}{2-\alpha}.
\end{align*}
By taking a maximum over $\alpha$, we obtain that the capacity is lower bounded with $\log\phi$.

$\hfill\blacksquare$

\section{Proof of Lemma \ref{lemma:beta_upper_bound}}\label{app:lemma_beta_1}
Before presenting the proof, we recall a known upper bound on the difference between two entropies for different PMFs is presented.
\begin{lemma}[Theorem $3$, \cite{igal}]\label{lemma:igal}
For two joint PMFs, $P$ and $Q$, on a finite set $\mathcal{X}$,
\begin{align*}
| H_P(X) - H_Q(X)|& \leq \| P_{X} - Q_{X}\|_{TV}\log(|\mathcal{X}|-1) + H_2(\| P_{X} - Q_{X}\|_{TV}).
\end{align*}
\end{lemma}

\begin{proof}[Proof of Lemma \ref{lemma:beta_upper_bound}]
For initial state $\xi$, the real distribution on $\mathcal{S}\times\mathcal{Q}$ at time $nD+k$ is $\xi T^{nD+k}$, while $\pi_k(\xi)$ is the distribution that was defined in \eqref{eq:D_limit}. Recall that the distribution $P_{Y,X|S,Q} = P_{Y|X,S}P_{X|S,Q}$ is determined by the policy and the channel. With some loss of accuracy, the dependence on the initial state $\xi$ might be omitted and the derivations hold for all $\xi$.

Consider the rewards difference at time $nD+k$:
\begin{align}\label{eq:trian}
|R(\xi T^{nD+k},f_\beta) - R(\pi_k(\xi),f_\beta)|& = |I_{T^{nD+k}}(X,S;Y|Q) - I_{\pi^k}(X,S;Y|Q)| \nn\\
& \stackrel{(a)}\leq  |H_{T^{nD+k}}(Y|Q) - H_{\pi^k}(Y|Q)| + |H_{T^{nD+k}}(Y|X,S,Q) - H_{\pi^k}(Y|X,S,Q)|,
\end{align}
where $(a)$ follows by the triangle's inequality. The first term in \eqref{eq:trian} can be bounded by
\begin{align}\label{eq:TV_first}
    |H_{T^{nD+k}}(Y|Q) - H_{\pi^k}(Y|Q)| &\stackrel{(a)}\leq |\mathcal{Q}| \max_q |H_{T^{nD+k}}(Y|Q=q) - H_{\pi^k}(Y|Q=q)| \nn\\
  & \stackrel{(b)}\leq |\mathcal{Q}| \max_q \left\{||T^{nD+k}_{Y|Q=q} - \pi^k_{Y|Q=q}||_{TV}\log(|\mathcal{Y}|-1)  + H_2(||T^{nD+k}_{Y|Q=q} - \pi^k_{Y|Q=q}||_{TV})\right\},
\end{align}
where $(a)$ follows by the triangle's inequality and $(b)$ follows from Lemma \ref{lemma:igal}.

Consider for all $q\in\mathcal{Q}$
\begin{align*}
   ||T^{nD+k}_{Y|Q=q} - \pi^k_{Y|Q=q}||_{TV} &\stackrel{(a)}\leq ||\xi T^{nD+k}_{Y,Q} - \pi^k_{Y,Q}(\xi)||_{TV}\nn\\
   &\stackrel{(b)}\leq  ||T^{nD+k}_{S,Q,X,Y} - \pi^k_{S,Q,X,Y}||_{TV} \nn\\
  &\stackrel{(b)}= ||T^{nD+k}_{S,Q} - \pi^k_{S,Q}||_{TV} \nn\\
  &\stackrel{(c)}\leq C\alpha^n,
\end{align*}
where $(a)$ follows by adding terms to the sum of total variation, $(b)$ follows from Lemma \ref{lemma:Paul} and $(c)$ follows from Lemma \ref{lemma:convergence_periodic}. Since $C\alpha^n\to 0$, there exists some $N'$ for which $||T^{nD+k}_{S,Q} - \pi^k_{S,Q}||_{TV} \leq 0.5$ for all $n>N'$.

Therefore, \eqref{eq:TV_first} can be bounded for all $n>N'$ as follows:
\begin{align*}
     |H_{T^{nD+k}}(Y|Q) - H_{\pi^k}(Y|Q)|&\stackrel{(a)}\leq |\mathcal{Q}| \left\{||T^{nD+k}_{S,Q} - \pi^k_{S,Q}||_{TV}\log(|\mathcal{Y}|-1)  + H_2(||T^{nD+k}_{S,Q} - \pi^k_{S,Q}||_{TV})\right\}\nn\\
     &\leq |\mathcal{Q}| \left\{C\alpha^n \log(|\mathcal{Y}|-1)  + H_2(C\alpha^n)\right\},
 \end{align*}
where $(a)$ follows from \eqref{eq:TV_first} and Lemma \ref{lemma:Paul}. The same derivation can be repeated for the second term in \eqref{eq:trian} resulting in the same convergence rate. To summarize, there exist some constants $C'>0$ and $\alpha\in(0,1)$ such that
\begin{align}\label{eq:trian_notation}
  |R(\xi T^{nD+k},f_\beta) - R(\pi_k(\xi),f_\beta)|&\leq C'\alpha^n + 2 H_2(C\alpha^n),
\end{align}
for all $n>N'$.

For all $\xi$, consider
\begin{align}\label{eq:lemma1_final}
 |\nu_\beta(\xi) -  \nu^\pi_\beta(\xi)| &= \left| \sum_{n=1}^\infty \beta^{nD} \sum_{k=0}^{D-1} \beta^k \left[ R(\xi T^{nD+k},f_\beta)  -  R(\pi_k(\xi),f_\beta)\right] \right|\nn\\
 &\stackrel{(a)}\leq \sum_{n=1}^\infty \beta^{nD} \sum_{k=0}^{D-1}\beta^k |R(\xi T^{nD+k},f_\beta)  - R(\pi_k(\xi),f_\beta)| \nn\\
 &\stackrel{(a)}\leq N'D\log|\mathcal{Y}| + \sum_{n=N'+1}^\infty \beta^{nD} \sum_{k=0}^{D-1}\beta^k |R(\xi T^{nD+k},f_\beta)  - R(\pi_k(\xi),f_\beta)| \nn\\
 &\stackrel{(b)}\leq N'D\log|\mathcal{Y}| + \sum_{n=N'+1}^\infty \beta^{nD} \sum_{k=0}^{D-1}\beta^k C'\alpha^n + 2 H_2(C\alpha^n) \nn\\
 &\stackrel{(c)}\leq N'D\log|\mathcal{Y}| + \sum_{n=1}^\infty C'\alpha^n + 2 H_2(C\alpha^n),
\end{align}
where $(a)$ follows by the triangle's inequality, $(b)$ follows from \eqref{eq:trian_notation} and $(c)$ follows from $\beta\leq 1$. Finally, by verifying that $\sum_{n=1}^\infty H_2(C\alpha^n) < \infty$ and by taking the supremum on both sides of \eqref{eq:lemma1_final} we have that $\sup_\beta |\nu_\beta(\xi) -  \nu^\pi_\beta(\xi)| < \infty$.
\end{proof}

\section{Proof of Lemma \ref{lemma:beta_upper_ast}}\label{app:lemma_beta_2}
The proof of Lemma \ref{lemma:beta_upper_ast} comprises two main steps. First, we derive an upper bound on $\nu^\pi_\beta(\xi)$ which does not depend on the initial state $\xi$. Secondly, we construct a new policy that can achieve a reward that is arbitrarily close to the provided upper bound. From the optimality of $f_\beta$, the two steps taken imply that all $\nu^\pi_\beta(\xi)$ are, indeed, close ``enough'' to the upper bound.
\begin{proof}[Proof of Lemma \ref{lemma:beta_upper_ast}]
Let us derive an upper bound on the average of $D$ consecutive rewards for some initial state $\xi$:
\begin{align}\label{eq:CaseA_upper}
  \frac{1}{D}\sum_{k=0}^{D-1} \beta^k R(\pi_k(\xi),f_\beta) &\stackrel{(a)}\leq   \frac{1}{D}\sum_{k=0}^{D-1} R(\pi_k(\xi),f_\beta)\nn\\
  &\stackrel{(b)}= \frac{1}{D}\sum_{k=0}^{D-1} I(X,S_k(\xi);Y|Q_k(\xi))\nn \\
               &\stackrel{(c)}= I(X,S;Y|Q,K(\xi)) \nn\\
               &\stackrel{(d)}\leq I_\xi(X,S;Y|Q)\nn\\
               &\stackrel{(e)}= R(\pi_{f_\beta},f_\beta),
\end{align}
where $(a)$ follows from $\beta\leq 1$, $(b)$ follows from the notation $p(s_k(\xi),q_k(\xi)) = \pi_k(\xi)$, $(c)$ follows by defining a uniform RV, $K$, on $[0:D-1]$ and $(d)$ follows from the fact that conditioning reduces entropy and from the Markov chain $Y-(X,S)-K(\xi)$, where the subscript $\xi$ is added to emphasize the dependence on the initial state. Finally, step $(e)$ shows that the marginal distribution does not depend on $\xi$; the marginal distribution of $p(s,q)$ for some $(s,q)\in A_i$ is
\begin{align*}
   p(s,q)&= \sum_k p(k,s,q)\\
         &= \frac{1}{D}\sum_k p(s_k(\xi),q_k(\xi))\\
         &\stackrel{(a)}= \frac{1}{D}\sum_k w_k(\xi) \pi(A_i)(s,q)\\
         &\stackrel{(b)} = \frac{1}{D}\pi(A_i)(s,q)\\
         &\stackrel{(c)}= \pi_{f_\beta}(s,q),
\end{align*}
where $(a)$ follows from \eqref{eq:D_limit}, $(b)$ follows from $\sum_k w_k(\xi)=1$ and the notation $\pi(A_i)(s,q)$ as the stationary distribution of the state $(s,q)$. Finally, $(c)$ follows from the property that in a periodic Markov chain each class has a uniform distribution.

The derivation above is used to provide an upper bound on $\nu^\pi_\beta(\xi)$:
\begin{align}\label{eq:lemma2_upper_ast}
    \nu^\pi_\beta(\xi) &=  \sum_{n=1}^\infty \beta^{nD} \sum_{k=0}^{D-1} \beta^k R(\pi_k(\xi),f_\beta) \nn\\
                       &=  \sum_{n=1}^\infty \beta^{nD} D \frac{1}{D} \sum_{k=0}^{D-1} \beta^k R(\pi_k(\xi),f_\beta)\nn\\
                       &\stackrel{(a)}\leq \sum_{n=1}^\infty \beta^{nD} D R(\pi_{f_\beta},f_\beta) \nn\\
                       &= \frac{D R(\pi_{f_\beta},f_\beta)}{1-\beta^D}\nn\\
                       &\stackrel{(b)}= \nu^\ast_\beta + R(\pi_{f_\beta},f_\beta) \frac{D - (1+\beta+\dots+\beta^{D-1})}{1-\beta^D} \nn\\
                       &\stackrel{(c)}= \nu^\ast_\beta + K_\beta,
\end{align}
where $(a)$ follows from \eqref{eq:CaseA_upper}, $(b)$ follows from the fact that $\nu^\ast_\beta = \frac{R(\pi_{f_\beta},f_\beta)}{1-\beta}$ and $(c)$ is just a notation $K_\beta$; by using L'Hopital's rule it can be noted that $\sup_\beta K_\beta<\infty$.

A new stationary policy, $f_\beta(\epsilon)$, is constructed by taking the policy $f_\beta$ and letting a path be with $\epsilon>0$ weights, so that the resultant graph is aperiodic. This modification is possible due to the aperiodicity assumption in Theorem \ref{theorem:main}. Moreover, $\epsilon$ is chosen to be small enough such that a node with modified outgoing edges still has positive probabilities for all other outgoing edges. The stationary distribution of this modified policy is denoted by $\pi(\epsilon)$, satisfying $\pi(0)=\pi_{f_\beta}$. The reward gained by the policy $f_\beta(\epsilon)$ is denoted by $\nu^\epsilon_\beta(\xi)$.

For $\epsilon\geq 0$, the stationary distribution exists and is unique, since it is a solution of linear equations. The stationary distribution is continuous with respect to $\epsilon$ since each entry in this vector is a rational function of $\epsilon$ and, clearly, $\epsilon=0$ is not a pole. We also know that mutual information is continuous w.r. to $\pi(\epsilon)$ and $f_\beta(\epsilon)$ and, therefore, the composition $I_{\pi(\epsilon)}\triangleq R(\pi(\epsilon),f_\beta(\epsilon))$ is continuous with respect to the parameter $\epsilon$.

By repeating the arguments in Lemma \ref{lemma:beta_upper_bound} with Lemma \ref{lemma:convergence} on the convergence rate of aperiodic Markov chains, it can be deduced that
\begin{align}\label{eq:new_policy_two_states}
 \sup_\beta |\nu^{\epsilon}_\beta(\xi) - \nu^{\epsilon}_\beta(\xi')| &< \infty,
\end{align}
for all $\xi,\xi'$. Note that \eqref{eq:new_policy_two_states} holds for all states and, specifically, for $\xi'=\pi(\epsilon)$.

For a fixed $\beta$, the continuity of each instantaneous reward in $\epsilon_\beta$ assures that there exists $\epsilon^\ast_\beta$ such that the difference between $|I_{\pi(0)}-I_{\pi(\epsilon_\beta)}|< 1-\beta$ for all $\epsilon_\beta<\epsilon^\ast_\beta$. By combining this continuity and \eqref{eq:new_policy_two_states}, we have
\begin{align}\label{eq:epsilon_continuity}
 \sup_\beta |\nu^{\epsilon^\ast_\beta}_\beta(\xi) - \nu^\ast| &< \infty,
\end{align}
for all $\xi$.

The reward $\nu^{\epsilon^\ast_\beta}_\beta(\xi)$ is achievable for initial state $\xi$ by the following trivial policy: for some initial state $\xi$ use policy $f_\beta(\epsilon)$ and, otherwise, use $f_\beta$. Clearly, the constructed policy does not change rewards for initial states other than $\xi$, and the optimality of $f_\beta$ gives that $\nu^{\epsilon^\ast_\beta}_\beta(\xi)\leq \nu_\beta(\xi)$.

For all initial states,
\begin{align}\label{eq:almost_final_lemma_2}
  \nu_\beta^\pi(\xi) &\stackrel{(a)}\leq  \nu_\beta^\ast + K_\beta \nn\\
  &\stackrel{(b)}= \nu^{\epsilon^\ast_\beta}_\beta(\xi) + K_{\beta,\epsilon^\ast_\beta}(\xi) + K_\beta \nn\\
  &\stackrel{(c)}\leq \nu_\beta(\xi) + K_{\beta,\epsilon^\ast_\beta}(\xi) + K_\beta \nn\\
  &\stackrel{(d)}= \nu^\pi_\beta(\xi) + K'_{\beta}(\xi) + K_{\beta,\epsilon^\ast_\beta}(\xi) + K_\beta,
\end{align}
where $(a)$ follows from inequality \eqref{eq:lemma2_upper_ast}, $(b)$ follows from \eqref{eq:epsilon_continuity} and the notation $K_{\beta,\epsilon^\ast_\beta}(\xi)$ for the difference between the rewards, $(c)$ follows by the optimaility of the policy $f_\beta$ and $(d)$ follows from Lemma \ref{lemma:beta_upper_bound} and the notation $\nu_\beta(\xi) = \nu^\pi_\beta(\xi) + K'_{\beta}(\xi)$, where $\sup_\beta |K'_{\beta,\epsilon}(\xi)| <\infty$.

By subtracting $\nu_\beta^\pi(\xi)$ from \eqref{eq:almost_final_lemma_2}, one can conclude that for all $\beta$
\begin{align}\label{eq:final_lemma_2}
  |\nu_\beta^\ast - \nu_\beta^\pi(\xi) |&\leq \max\{ |K'_{\beta}(\xi) + K_{\beta,\epsilon^\ast_\beta}(\xi)| , K_\beta \}.
\end{align}
By taking a supremum on both sides of \eqref{eq:final_lemma_2}, the proof of Lemma \ref{lemma:beta_upper_ast} is concluded.
\end{proof}

\bibliography{ref}
\bibliographystyle{IEEEtran}

\end{document}